\documentclass[11pt,fleqn]{article}

\usepackage{subcaption}

\usepackage{amsfonts}
\usepackage{amsmath,amssymb}
\usepackage{amsthm}
\usepackage{algorithmic}
\usepackage{array}
\usepackage{epsf}
\usepackage{graphicx,psfrag,color,pst-grad,pstcol}
\graphicspath{ {./figures/} }
\usepackage{latexsym}
\usepackage{calc}
\usepackage{multicol}
\usepackage{float}

\def\qed{\hfill \mbox{\rule[0pt]{1.5ex}{1.5ex}} \vskip.1cm}
\def\R{\mathbb{R}}
\def\cG{{\cal G}}
\def\cL{{\cal L}}

\def\cT{{\cal T}}

\def\cB{{\cal B}}

\def\cE{{\cal E}}

\def\cP{{\cal P}}

\def\cM{{\cal M}}

\def\cQ{{\cal Q}}

\def\qed{\hfill \mbox{\rule[0pt]{1.5ex}{1.5ex}}\vskip.3cm}

\newtheorem{remark}{Remark}
\newtheorem{theorem}{Theorem}
\newtheorem{prop}{Proposition}
\newtheorem{assumption}{Assumption}

\def\BibTeX{{\rm B\kern-.05em{\sc i\kern-.025em b}\kern-.08em
    T\kern-.1667em\lower.7ex\hbox{E}\kern-.125emX}}

\title{\LARGE \bf Mixed Potential Approach to Convergence of Nonlinear RLC Circuits with Memristors}

\author{Mauro Di Marco, Mauro Forti, Luca Pancioni,\\ Giacomo Innocenti, Alberto Tesi
\thanks{M. Di Marco, M. Forti and L. Pancioni are with the Department of Information Engineering and Mathematics, University of Siena, v. Roma 56 - 53100 Siena, Italy, e--mail: {\tt\small mauro.dimarco@unisi.it, mauro.forti@unisi.it, luca.pancioni@unisi.it}.
G. Innocenti and A. Tesi are with the Department of Information Engineering,
University of Florence, via S. Marta 3 - 50139 Firenze, Italy, e--mail: {\tt\small giacomo.innocenti@unifi.it,alberto.tesi@unifi.it}.
}
}

\begin{document}

\maketitle

\begin{abstract}
The paper considers a large class of nonlinear circuits, termed RLCM, containing all four basic
circuit elements, i.e., resistors, inductors, capacitors and
memristors. A companion paper \cite{DiMarco2026BraytonTheor} has
introduced a \emph{mixed potential for RLCM circuits} generalizing that found by Brayton and Moser for circuits without memristors. In this paper, systematic
Lyapunov-like results on convergence of RLCM circuits are proved by means of the mixed
potential. These hold under the basic assumption that an RLCM circuit has a complete set
of variables in the flux-charge domain and they require, roughly speaking, that
there is a balance, which is quantitatively estimated, between capacitors and inductors. The
convergence results
are robust with respect to circuit parameter variations and they include cases where the
memristor circuits possess multiple stable equilibrium points, which is of importance for instance to
implement content addressable memories (CAMs). The results extend to circuits possessing
all four basic circuit elements previous results that pertain to circuits without
memristors or memristor circuits without inductors.
The main proofs are conducted by using the flux-charge analysis method (FCAM) to analyze RLCM
circuits in the flux-charge domain.
\end{abstract}


\ \vskip1.cm

\section{Introduction}
\label{sect:intro}

Memristor has been introduced by L. Chua in a classical paper published in 1971 as
the fourth basic passive circuit element together with resistor, inductor and
capacitor~\cite{Chua1971}.
Memristor circuits are gaining an increasing importance in modern electronics for their
peculiar ability to mimic the processing capabilities of nervous systems \cite{huang2021editorial,Sirakoulis2022717,Wang2023AEM,Duan2024}. Memristors are
used to model neurons, since they are able to memorize their state without
using a power supply \cite{boybat2018neuromorphic,yang2022ResProgMem}. Moreover, they are tailor made to implement synaptic-like
connections by means of their state-dependent conductance. \emph{In-memory computing} schemes
have been devised in the literature where the processing and memorization are
at the same physical level in the memristors~ \cite{ielmini2020device,ascoli2020theoretical,
10518003,SpecIssue-InMemory2023,11053213,Xiao20242294TCASI,Yang20254127}. These schemes appear to be effective
to overcome the so-called memory wall and Von-Neumann bottleneck of
traditional computing machines.

One of the most important dynamical properties of nonlinear circuits and
neural networks (NNs) is convergence towards equilibrium points (EPs) \cite{chua1980dynamic,CG83,Hop84}.
In fact, nonlinear circuits can burst into undesirable oscillations,
so that finding design conditions ensuring convergence is of practical relevance.
Convergence is known to be of paramount interest in nonlinear circuits modeling
NNs. The main reason is that convergent NNs with
multiple stable EPs, also named multistable NNs, are tailor made to
implement content addressable
memories (CAMs) and to solve combinatorial optimization problems and
several other signal processing tasks in real time \cite{CG83,Hop84,CY88a}.

A huge and comprehensive body
of relevant contributions are available in the literature on convergence of
nonlinear circuits and NNs \emph{without memristors},
see, e.g.,~\cite{chua1980dynamic,brayton1971nonlinear,TNN-GAS-review,Liu20231098} and references therein.
Some main results
rely on Lyapunov approach,
monotone dynamic system theory for
cooperative systems~\cite{DFGP11d,DFGP11b} and the global consistency of decisions for competitive
NNs~\cite{Gro80}. The Nobel prize
awarded to J.\ J.\ Hopfield for the celebrated NN carrying his name has
recently given a significant boost in studying convergence of NNs.
In contrast, the study of convergence for nonlinear circuits and NNs \emph{with
memristors} is still in its infancy.
One reason is that only recently memristors have attracted
the attention of the scientific community. Another reason is that
analyzing circuits with memristors appears to be more complex due to the
elusive dynamic properties displayed by these elements. As a matter of fact,
the study of convergence when nonlinear circuits or NNs contain
emerging electronic elements as memristors is deemed to be a very
significant research topic. The importance is even more
evident when considering that, as it was noticed before,
the modern trend is to design NNs where both the interconnections
and the neurons are implemented with memristors.

Only a few results on convergence of memristor circuits are available
in the literature. A large part of them pertain to memristor circuits with a
neural architecture where use is made of the ideal memristor
introduced by L. Chua and some of its variants as the celebrated HP memristor~\cite{DFP16,di2017memristor,DiMarco2020LQprog,10144925,DiMarco2022rectifying,di2022convergence}. Another line of research concerns
convergence of NNs where memristors are modeled by switching resistors~\cite{WZ12,GWY2013,YCY14,ZSYS13,WWZ12}.
However, the relation of that model with the classical Chua's memristor
has not been investigated as yet.
A recent article has addressed convergence for a class of nonlinear
circuits, termed RCM, containing \emph{three basic circuit elements},
i.e., resistors, capacitors and memristors~\cite{di2025robust}.
It is proved that convergence
is robust within that class, i.e., it holds for any value of the circuit
parameters and memristor nonlinearities involved.
One basic limitation is that none of the quoted results pertains to
memristor circuits including \emph{both capacitors and inductors}.

In this paper, we consider a class of nonlinear circuits containing \emph{all four
basic circuit elements}, i.e., resistors, capacitors, inductors and
memristors.
The class, termed RLCM, generalizes that studied in~\cite{di2025robust}
where inductors are not included. Considering also inductors is relevant from
a theoretic as well as a practical viewpoint. In fact, we can expect a richer dynamical behavior,
including convergence, oscillatory and chaotic phenomena, when both the electric
and magnetic fields are simultaneously present in a nonlinear circuit. Inductors can be deliberately
inserted in a circuit to obtain particular types of dynamics. On the other hand, it
is well-known that any serious model of circuits operating at high frequencies needs to
necessarily account for parasitic inductors.

A companion paper has introduced a \emph{mixed potential} in the flux-charge domain (FCD)
for RLCM circuits~\cite{DiMarco2026BraytonTheor}. The mixed
potential extends that introduced by Brayton and Moser in the traditional
voltage-current domain (VCD) for memristor-less RLC circuits.
This paper uses the mixed potential in~\cite{DiMarco2026BraytonTheor}, in combination with the flux-charge
analysis method (FCAM)~\cite{Corinto-Forti-I,cfc2020}, to prove systematic
Lyapunov-type results on convergence for nonlinear RLCM circuits.
The main contributions can be summarized as follows:

(a) We establish easily checkable conditions ensuring convergence for relevant
subclasses of RLCM circuits.
Typical results state
that convergence is guaranteed provided a given matrix, depending on the
circuit topology and on capacitances and inductances, has a norm bounded
by a suitable quantity. More physically, this means that there is convergence
provided there is a balance between the inductors and capacitors in the
circuit.

(b) The convergence results are robust, i.e., they hold
on an open set of circuit parameters and for a large class
of memristor nonlinearities.
Moreover, they include the most
interesting case where the memristor circuits possess \emph{multiple stable EPs},
which is a key property
to implement for instance a CAM
or to use the circuits for solving combinatorial optimization problems or
other signal processing tasks in real time.

(c) Convergence is established both in the FCD and in the VCD.
For the considered subclasses, it is shown that additional dynamical
aspects can be singled out. These include the presence of invariants
of motion for the equations in the VCD and the possibility to
obtain a foliation of the state space in the VCD in invariant manifolds
where the circuits obey a reduced-order dynamics. Moreover, links are
established between the invariants of motion and the mixed potential
enabling to study how the convergence properties depend upon the
manifolds.

(d) The results are illustrated via the simulations of selected
RLCM circuits.

The obtained conditions for convergence are, to the authors knowledge, the only existing ones
that can be applied to nonlinear circuits that contain all four basic circuit elements. The results
are an extension of the previously quoted ones for nonlinear circuits, NNs with ideal
memristors and NNs with memristors modeled as switching devices, where once
more inductors are not considered.

The paper is organized as follows. The considered class of RLCM circuits is described in
Sect.\ \ref{sect:class_RLCM} and the SE description of RLCM circuits is provided in
Sect.\ \ref{sect:FCD}. Section\ \ref{sect:subclasses} considers some interesting
subclasses of RLCM circuits. The main results on convergence are established in
Sect.\ \ref{sect:RCM conv} and Sect.\ \ref{sect:convRLCM} and their significance is
discussed in Sect.\ \ref{sect:disc}. Examples illustrating the convergence results are
worked out in Sect.\ \ref{sect:ex} while the main conclusions are
collected in Sect.\ \ref{sect:concl}.

\section{Class RLCM of Memristor Circuits}
\label{sect:class_RLCM}

We consider a class of nonlinear circuits, denoted RLCM, containing:
\begin{itemize}
   \item[$\bullet$] $n_C$ passive capacitors $i_{Ci}=C_i \dot v_{Ci}$, where $C_i>0$, and $n_L$ passive inductors $v_{Li}=L_i \dot i_{Li}$, where $L_i>0$;
   \item[$\bullet$] $n_R$ resistors $v_{Ri}=R_i i_{R i}$; each resistor may be \emph{passive}, i.e., $R_i>0$, or \emph{active}, i.e., $R_i<0$;
   \item[$\bullet$] $n_{\Phi}$ ideal flux-controlled memristors $M_{\Phi i}$ with constitutive relations (CRs) $Q_{Mi}= \hat Q_{Mi}(\Phi_{Mi})$ \cite{Chua1971}, where
  $\Phi_{Mi}(t)=\int_{-\infty}^t v_{Mi}(\sigma)d\sigma$ is the flux (or voltage momentum) and
  $Q_{Mi}(t)=\int_{-\infty}^t i_{Mi}(\sigma)d\sigma$ is the charge (or current momentum). Moreover,
  $\hat Q_{Mi}: \R \to \R$ is the nonlinear memristor characteristic;
  \item[$\bullet$] $n_{Q}$ ideal charge-controlled memristors $M_{Qi}$ with CRs $\Phi_{Mi}= \hat \Phi_{Mi}(Q_{Mi})$, where
   $\hat \Phi_{Mi}: \R \to \R$ is the nonlinear memristor characteristic.
\end{itemize}

It is noticed that independent current or voltage sources are not explicitly considered.
As discussed in \cite{DiMarco2026BraytonTheor}, this is not restrictive, indeed those sources
are typically used to set initial conditions at $t=0$ for capacitors ($v_{c0i}=v_{Ci}(0)$), inductors
($i_{Li0}=i_{Li}(0)$) and memristors ($\Phi_{M0i}=\Phi_{Mi}(0)$ and $Q_{Mi0}=Q_{Mi}(0)$) but for $t \ge 0$ we may suppose that they vanish.

For the flux-controlled memristors, we suppose henceforth that the next assumptions hold.

\begin{assumption}
\label{assu:memp}
We have that:

1) $\hat Q_{Mi} \in C^1(\R)$ and $\hat Q_{Mi}(0)=0$;

2) $\hat Q_{Mi}'(\Phi_{Mi}) \ge 0$ for any $\Phi_{Mi}$;

3) there exist $G_{mi}>0$, $\tilde \Phi_{Mi} >0$ such that
\begin{equation}\label{Gmi}
    \frac{\hat Q_{Mi}(\Phi_{Mi})}{\Phi_{Mi}} \ge G_{mi}>0, \ \ |\Phi_{Mi}|> \tilde \Phi_{Mi}.
\end{equation}
\end{assumption}

A memristor implemented via
for instance a typical metal-insulator-metal (MIM) structure is a passive element since there are no internal sources of power. Passivity is equivalent to $\hat Q_{Mi}$ is a monotone non-decreasing function, which is guaranteed by condition 2) of Assumption\ \ref{assu:memp}. Dually, for a
a charge-controlled memristor, we assume the following.

\begin{assumption}
\label{assu:memc}
We have that:

1) $\hat \Phi_{Mi} \in C^1(\R)$ and $\hat \Phi_{Mi}(0)=0$;

2) $\hat \Phi_{Mi}'(Q_{Mi}) \ge 0$ for any $Q_{Mi}$;

3) there exist $R_{mi}>0$, $\tilde Q_{Mi} >0$ such that
\begin{equation}\label{Rmi}
    \frac{\hat \Phi_{Mi}(Q_{Mi})}{Q_{Mi}} \ge R_{mi}>0, \ \ |Q_{Mi}|> \tilde Q_{Mi}.
    \end{equation}
\end{assumption}


All ideal passive memristors considered in the literature satisfy these assumptions. This
is true for instance in the case of polynomial-type memristors as the cubic memristor
$Q_{Mi}=\hat Q_{Mi}(\Phi_{Mi})=a_1 \Phi_{Mi}+a_3 \Phi_{Mi}^3$, where $a_1 \ge 0$, $a_3>0$, or $C^1$ approximations
of piecewise-linear memristors $Q_{Mi}=\hat Q_{Mi}(\Phi_{Mi})=b \Phi_{Mi} +(1/2)(a-b)(|\Phi_{Mi}+\Phi_\sigma|-|\Phi_{Mi}-\Phi_\sigma|)$, where $a,b, \Phi_\sigma > 0$ \cite{Chua2015}. Also, the celebrated HP memristor \cite{Williams2008}, or a combination of them,
can be brought back to the considered model via a suitable transformation.

\section{SEs in the FCD and Mixed Potential}
\label{sect:FCD}

This section summarizes some main results obtained in \cite{DiMarco2026BraytonTheor} which are needed to study convergence towards EPs of an RLCM memristor circuit $\mathfrak{N}$. First, we recall
some basic facts about FCAM, a technique which is used throughout the analysis in the paper. Then,
we discuss the main assumptions
and procedure used in \cite{DiMarco2026BraytonTheor} to write the state equations (SEs) in the FCD of $\mathfrak{N}$, the concept of the mixed potential of $\mathfrak{N}$ and its relationships with the SE formulation.

\subsection{FCAM and Main Topological Assumptions}
\label{sect:FCAM}

FCAM is based on analyzing an RLCM circuit $\mathfrak{N}$ in the flux-charge domain (FCD) via: (a) the constitutive relations (CRs) of circuit elements
in terms of the incremental flux $ \varphi(t)=\int_{0}^t v(\sigma)d\sigma$ and incremental charge
$q(t)=\int_0^t i(\sigma)d\sigma$; (b)  Kirchhoff-flux-law (K$\varphi$L) and
Kirchhoff-charge-law (K$q$L) for incremental quantities. Henceforth, for simplicity, we omit the adjective incremental if there is no ambiguity.

It is shown in \cite[Sect.\ 2C]{DiMarco2026BraytonTheor} that, if we consider the analogies $\varphi \sim v$ and $ q \sim i$,
then an RLCM circuit $\mathfrak{N}$ in the FCD is the analogous of a nonlinear RLC circuit in the VCD. In particular, a capacitor (resp., an inductor) in the FCD is the analogous of a capacitor (resp., an inductor) in the VCD.
Moreover, a flux-controlled (resp., charge-controlled) memristor in the FCD is the equivalent
of a nonlinear initial-condition dependent voltage-controlled (resp., charge-controlled) resistor in the VCD.


We suppose as in \cite{DiMarco2026BraytonTheor} that the following hypotheses are enforced.

\begin{figure}[ht]
\begin{center}
\begin{subfigure}{0.8\columnwidth}
  \includegraphics[width=\linewidth]{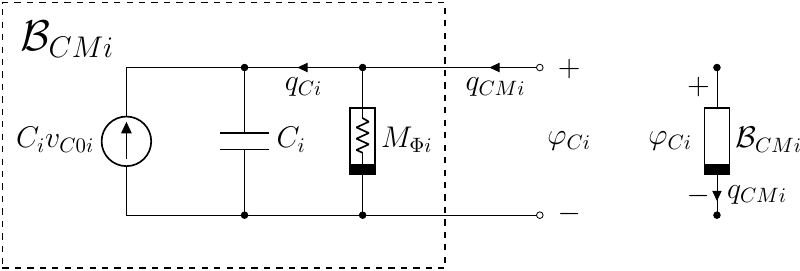}
  \caption{\small \small }
\end{subfigure}\\
\begin{subfigure}{0.8\columnwidth}
  \includegraphics[width=\linewidth]{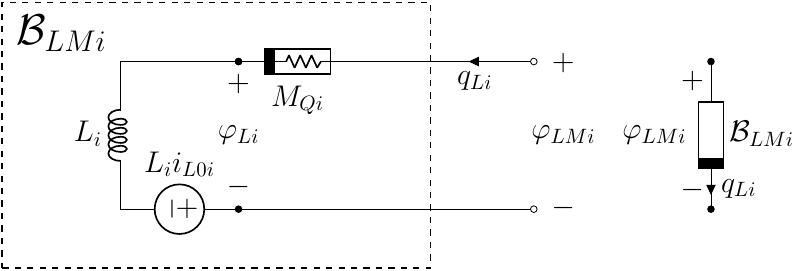}
  \caption{\small \small }
\end{subfigure}
\caption{\small (a) Composite element $\cB_{CMi}$ given by
a capacitor in parallel to a flux-controlled memristor and (b) composite
element $\cB_{LMi}$ given by an inductor in series with
a charge-controlled memristor.}
\label{fig:composite}
\end{center}
\end{figure}

\begin{assumption}
\label{assu:case 2 RLCM 1}
Each flux-controlled memristor is in parallel to a capacitor and each charge-controlled memristor is in series with an inductor.
\end{assumption}

On the basis of Assumption\ \ref{assu:case 2 RLCM 1}, we can suppose henceforth that $\mathfrak{N}$ contains the following two-terminal elements:

(i) $n_\Phi$ two-terminal elements $\cB_{CMi}$, $i=1,\dots,n_\Phi$, where $0 \leq n_\phi \le n_C$,
constituted by a capacitor
in parallel to a flux-controlled memristor (cf.\ Fig.\ \ref{fig:composite}(a)) satisfying
the CR
$$
q_{CMi}=C_i \dot \varphi_{Ci}+\hat Q_{Mi}(\varphi_{Ci}+\Phi_{M0i})-\hat Q_{Mi}(\Phi_{M0i})-C_iv_{C0i}.
$$

(ii) $n_C-n_\Phi$ elements $\cB_{Ci}$,
$i=n_{\Phi}+1,\dots,n_C$,
given by a capacitor without a memristor in parallel obeying the CR
\begin{equation}\label{Ci FCD}
    C_i \dot \varphi_{Ci}= q_{Ci}+C_iv_{C0i}
\end{equation}

(iii) $n_Q$ two-terminal elements termed $\cB_{LMi}$, $i=1,\dots,n_Q$, where $0 \leq n_Q \le n_L$, given by an inductor in series with a charge-controlled memristor (cf.\ Fig.\ \ref{fig:composite}(b))
satisfying
$$
\varphi_{LMi}=L_i \dot q_{Li}+\hat \Phi_{Mi}(q_{Li}+Q_{M0i})-\hat \Phi_{Mi}(Q_{M0i})-L_ii_{L0i}
$$

(iv) $n_L-n_Q$ elements $\cB_{Li}$,
$i=n_Q+1,\dots,n_L$, given by an inductor without a memristor
in series with CR
\begin{equation}\label{Li FCD}
    L_i \dot q_{Li}= \varphi_{Li}+L_i i_{L0i}.
\end{equation}

(v) $n_R$ resistors with CR $\varphi_{Ri}=R_i q_{R i}$.

It is worth to note that the CRs of circuit elements depend upon the initial conditions
$v_{C0i},i_{L0i},\Phi_{M0i}$ and $Q_{M0i}$ in the VCD.

Consider $\mathfrak{N}$, an associated digraph $\cG$, a tree $\cT$ and a
corresponding co-tree $\cL$ of $\cG$.

\begin{assumption}
\label{assu:case 2 RLCM 1b}
There is a tree $\cT$ and a corresponding co-tree $\cL$ such that:

1) all elements $\cB_{CM}$ and $\cB_C$ belong to $\cT$ and the remaining elements of $\cT$ are resistors;

2) all elements $\cB_{LM}$ and $\cB_L$ belong to $\cL$ and the remaining elements of $\cL$ are resistors;

3) each resistor in $\cL$ forms a loop exclusively with elements $\cB_{CM}$ and/or $\cB_C$.
\end{assumption}

\begin{remark}
Assumption\ \ref{assu:case 2 RLCM 1b}
guarantees that the capacitor fluxes $(\varphi_{C1},\dots, \varphi_{CnC})^\top \in \R^{n_C}$
and the inductor charges $(q_{L1},\dots, q_{LnL})^\top \in \R^{n_L}$ are a complete set of variables for
$\mathfrak{N}$ in the FCD (cf.\ \cite[Def.\ 1]{DiMarco2026BraytonTheor}).
\end{remark}

\subsection{SEs in the FCD}
\label{sect:SEs}

Consider an RLCM circuit $\mathfrak{N}$ and an associated digraph
$\cG$.
Suppose without loss of generality that $\cG$ is connected and
there are $b$ branches and $n$ nodes.
Let $\cT$ be a tree and $\cL$ the corresponding co-tree of $\cG$ chosen in a way that
Assumption\ \ref{assu:case 2 RLCM 1b} is satisfied.
Denote by $\varphi_\cT \in \R^{n-1}$, $q_\cT \in \R^{n-1}$ the tree fluxes and charges, respectively
and by $\varphi_\cL \in \R^{b-n+1}$, $q_\cL \in \R^{b-n+1}$ the co-tree fluxes and charges, respectively.
K$q$Ls for the fundamental cut-sets can be written in matrix-vector notation as \cite{chua-book}
\begin{equation}\label{kQL}
      q_\cT + A  q_\cL = 0
\end{equation}
where $A \in \R^{(n-1) \times (b-n+1)}$ is a topological matrix with elements $\{ -1,0,1 \}$. Decompose matrix $A$ as follows
\begin{equation}\label{P partition}
A=    \begin{pmatrix}
      A_{CL} & A_{CR} \\
      A_{RL} & A_{RR} \\
    \end{pmatrix}
\end{equation}
where $A_{CL}$ and $A_{CR}$ account for the contribution to the charges of elements $\cB_{CM}$ and $\cB_C$, belonging to $\cT$, due to the elements of $\cL$ (inductors and resistors, respectively). Similarly, $A_{RL}$ and $A_{RR}$ account for the contribution to the charges of the resistors belonging to $\cT$ due to the elements of $\cL$. Under Assumption\ \ref{assu:case 2 RLCM 1b}, it can be shown that $A_{RR}=0$
\cite[Prop.\ 2]{DiMarco2026BraytonTheor}.

Now, define matrix $H \in \R^{(n_C+n_L) \times (n_C+n_L)}$ as
\begin{equation}\label{H cas2}
    H=\begin{pmatrix}
        H_{\alpha \alpha} & H_{\alpha \beta} \\
        H_{\beta \alpha} & H_{\beta \beta} \\
      \end{pmatrix}=
      \begin{pmatrix}
        A_{CR} (R_{\cL})^{-1} A_{CR}^\top & A_{CL} \\
        -A_{CL}^\top & A_{RL}^\top R_{\cT} A_{RL} \\
      \end{pmatrix}
\end{equation}
where $R_{\cL}$ (resp., $R_{\cT}$) is a diagonal matrix whose entries are the resistances of the resistors belonging to $\cL$ (resp., $\cT$).

Let us introduce these notations:
\begin{itemize}
\item[$\bullet$] $\varphi_C=(\varphi_{Ci})_{i=1,\dots,n_C} \in \R^{n_C}$
\item[$\bullet$] $\varphi_{CM}=(\varphi_{Ci})_{i=1, \dots,n_\Phi} \in \R^{n_\Phi}$
\item[$\bullet$] $\hat Q_{MC}(\cdot)=(\hat{Q}_{Mi}(\cdot))_{i=1, \dots,n_\Phi}: \R^{n_\Phi} \to \R^{n_\Phi}$
\item[$\bullet$] $\Phi_{M0}=(\Phi_{M0i})_{i=1, \dots,n_\Phi} \in \R^{n_\Phi}$
\item[$\bullet$] $v_{C0}=(v_{C0i})_{i=1,\dots,n_C} \in \R^{n_C}$
\item[$\bullet$] $C={\rm diag }(C_1,\dots,C_{nC}) \in \R^{n_C \times n_C}$
\item[$\bullet$] $q_L=(q_{Li})_{i=1,\dots,n_L} \in \R^{n_L}$
\item[$\bullet$]  $q_{LM}=(q_{Li})_{i=1, \dots,n_Q} \in \R^{n_Q}$
\item[$\bullet$] $\hat \Phi_{ML}(\cdot)=(\hat \Phi_{Mi}(\cdot))_{i=1, \dots,n_Q}: \R^{n_Q} \to \R^{n_Q}$
\item[$\bullet$] $Q_{M0}=(Q_{M0i})_{i=1, \dots,n_Q} \in \R^{n_Q}$
\item[$\bullet$] $i_{L0}=(i_{L0i})_{i=1,\dots,n_L} \in \R^{n_L}$
\item[$\bullet$] $L={\rm diag }(L_1,\dots,L_{nL}) \in \R^{n_L \times n_L}$.
\end{itemize}

Under Assumptions\ \ref{assu:case 2 RLCM 1}, \ref{assu:case 2 RLCM 1b}, the SEs describing the behavior of $\mathfrak{N}$ in the FCD
are obtained as {\cite{DiMarco2026BraytonTheor}
\begin{align}\label{SEsFCDRLCM}
\begin{split}
    \begin{pmatrix}
      C \dot \varphi_C \\
      L \dot q_L \\
    \end{pmatrix}
=&
-H
\begin{pmatrix}
       \varphi_C \\
       q_L \\
    \end{pmatrix}\\
 &
        -\begin{pmatrix}
       \hat Q_{MC}( \varphi_{CM}+\Phi_{M0})-
       \hat Q_{MC}(\Phi_{M0})\\
       0 \\
       \hat \Phi_{ML}( q_{LM}+Q_{M0})-
       \hat \Phi_{ML}(Q_{M0})\\
       0 \\
     \end{pmatrix}\\
     &+
     \begin{pmatrix}
      C v_{C0} \\
      L i_{L0} \\
    \end{pmatrix}.
\end{split}
\end{align}
It is worth to notice that the
vector field defining the SEs depends upon the initial conditions
\begin{equation}\label{w0}
    w_0=(v_{C0}, i_{L0}, \Phi_{M0},Q_{M0}) \in \R^{n_C+n_L+n_\Phi+n_Q}
\end{equation}
of the state variables in the VCD.


The SEs of $\mathfrak{N}$ in the VCD can be obtained by differentiating (\ref{SEsFCDRLCM}),
namely,
\begin{align}\label{SEs VCD RLCM}
\begin{split}
    \begin{pmatrix}
      C \dot v_{CM} \\
      L \dot i_{LM} \\
    \end{pmatrix}
=&
-H
\begin{pmatrix}
       v_C \\
       i_L \\
    \end{pmatrix}
        -\begin{pmatrix}
       \hat Q'_{MC}( \Phi_{M})v_C\\
       0 \\
       \hat \Phi'_{ML}( Q_{M})i_L\\
       0 \\
     \end{pmatrix}
     \\
     \dot \Phi_{M}=&v_{CM}\\
     \dot Q_{M}=&i_{LM}.
\end{split}
\end{align}

\subsection{Mixed Potential}
\label{sect:mixed}

It is shown in \cite[Th.\ 1]{DiMarco2026BraytonTheor} that, under Assumptions \ref{assu:case 2 RLCM 1}, \ref{assu:case 2 RLCM 1b},
it is possible to define the following mixed potential
$\cP(\varphi_C,q_L;w_0):
\R^{n_C+n_L} \to \R$ for an RLCM circuit $\mathfrak{N}$
\begin{align}\label{mixed n1}
\begin{split}
    &\cP(\varphi_C,q_L;w_0)=
    \frac{1}{2}  \varphi_C^\top H_{\alpha \alpha}  \varphi_C
    -\frac{1}{2}  q_L^\top H_{\beta \beta}  q_L
    +  \varphi_C^\top H_{\alpha \beta}  q_L\\
    &+\sum_{i=1}^{n_{\Phi}} \int_0^{ \varphi_{\alpha i}}(\hat Q_{Mi}(\rho+\Phi_{M0i})-
       \hat Q_{Mi}(\Phi_{M0i}))d\rho \\
    &-\sum_{j=1}^{n_{Q}} \int_0^{ q_{\beta j}}
(\hat \Phi_{Mj}(\rho+Q_{M0j})-
       \hat \Phi_{Mj}(Q_{M0j}))d\rho \\
    &- \varphi_C^\top C v_{C0}
+  q_L^\top L i_{L0}.
\end{split}
\end{align}
Note that the mixed potential depends upon $w_0$, i.e., the initial conditions for the
state variables in the VCD (cf.\ (\ref{w0})).

One key property is that, via the mixed potential, the SEs
(\ref{SEsFCDRLCM}) of $\mathfrak{N}$ in the FCD can be compactly
written in the form
\begin{equation}\label{SEs FCD RLCM mixed}
     \begin{pmatrix}
      C \dot \varphi_C \\
      L \dot q_L \\
    \end{pmatrix}
=    \begin{pmatrix}
      -\frac{\partial}{\partial  \varphi_C} \cP( \varphi_C, q_L;w_0) \\
      \frac{\partial}{\partial  q_L} \cP( \varphi_C, q_L;w_0) \\
    \end{pmatrix}.
\end{equation}

\begin{remark}
\label{rem:recipro}
The existence of a mixed potential for $\mathfrak{N}$ is related to the reciprocity
of a multi-port where the elements $\cB_{CM}$, $\cB_C$, $\cB_{LM}$ and $\cB_L$
are connected (cf.\ \cite{DiMarco2026BraytonTheor}). This in turn is equivalent to the property
that, for structural reasons, matrix $H$ in (\ref{H cas2}) has a special
structure where $H_{\alpha \alpha}$ and $H_{\beta \beta}$ are symmetric matrices,
whereas we have $H_{\beta \alpha}=-H_{\alpha \beta}^\top$.
\end{remark}

\begin{remark}
The explicit expression of the mixed potential in (\ref{mixed n1}) will
play a crucial role to prove Lyapunov-like
results on convergence for RLCM circuits in Sect.\ \ref{sect:convRLCM}.
\end{remark}

\section{Relevant Subclasses of RLCM}
\label{sect:subclasses}

In this section, we single out three relevant subclasses of RLCM circuits for which
we will be able to prove systematic convergence results (Sect.\ \ref{sect:convRLCM}).

\subsection{Subclass RLCM$a$}
\label{sect:subc RLCMa}

Suppose that each capacitor has in parallel a
flux-controlled memristor and there is no other memristor. We term
this subclass RLCM$a$.
In this case, we have $n_{\Phi}=n_C$, $\varphi_{CM}= \varphi_C$ and $\Phi_M=\varphi_C+\Phi_{M0}$.
The SEs of $\mathfrak{N}$ in the FCD are
\begin{align}\label{SEs FCD RLCM case 1}
\begin{split}
      \begin{pmatrix}
      C \dot \varphi_C \\
      L \dot q_L \\
    \end{pmatrix}
=& -H
\begin{pmatrix}
       \varphi_C \\
       q_L \\
    \end{pmatrix}\\
    &-\begin{pmatrix}
       \hat Q_{MC}( \varphi_C+\Phi_{M0})-
       \hat Q_{MC}(\Phi_{M0})\\
       0 \\
    \end{pmatrix}\\
     &+
     \begin{pmatrix}
      C v_{C0} \\
      L i_{L0} \\
    \end{pmatrix}
\end{split}
\end{align}
where $H$ is as in (\ref{H cas2}).

Suppose that sub-matrix $H_{\beta \beta}$ is nonsingular and consider
function $\cQ:\R^{2n_C+n_L}
    \to \R^{n_C}$ given by
\begin{align}\label{Q(t) case 1}
\begin{split}
   \cQ=&\cQ(v_C,i_L,\Phi_M)\\
   =&(H_{\alpha \alpha}-H_{\alpha \beta}H_{\beta \beta}^{-1}(-H_{\alpha \beta}^\top))\Phi_M
    +\hat Q_{MC}(\Phi_M)\\
    &+Cv_C-H_{\alpha \beta} H_{\beta \beta}^{-1} Li_L.
    \end{split}
\end{align}
The change of variables $x=\Phi_M=\varphi_C+\Phi_{M0}$,
$y=q_L-H_{\beta \beta}^{-1}(-H_{\alpha \beta}^{\top} \Phi_{M0}+Li_{L0})$ yields
the following form for the SEs in the FCD, which will be useful for the analysis that follows
\begin{align}\label{SEs FCD RLCM case 1 xy}
\begin{split}
    \begin{pmatrix}
      C \dot x \\
      L \dot y \\
    \end{pmatrix}
=& -H
\begin{pmatrix}
       x \\
       y \\
    \end{pmatrix}
        -\begin{pmatrix}
       \hat Q_{MC}(x)\\
       0 \\
    \end{pmatrix}
     +
     \begin{pmatrix}
      Q_0 \\
      0 \\
    \end{pmatrix}
\end{split}
\end{align}
where
\begin{align}\label{Q0 case 1}
\begin{split}
    Q_0 =\cQ(0)=&(H_{\alpha \alpha}-H_{\alpha \beta}H_{\beta \beta}^{-1}(-H_{\alpha \beta}^\top))\Phi_{M0}\\
    &+\hat Q_{MC}(\Phi_{M0})+Cv_{C0}-H_{\alpha \beta} H_{\beta \beta}^{-1} Li_{L0}
\end{split}
\end{align}
is a constant term depending upon the initial conditions $w_0$ for the state variables in the VCD.

Differentiating (\ref{SEs FCD RLCM case 1}), or (\ref{SEs FCD RLCM case 1 xy}), we obtain
that the SEs of $\mathfrak{N}$ in the VCD are given by
\begin{align}\label{SEs VCD RLCM case 1}
\begin{split}
        \begin{pmatrix}
      C \dot v_C \\
      L \dot i_L \\
    \end{pmatrix}
=&
-H
\begin{pmatrix}
       v_C \\
       i_L \\
    \end{pmatrix}
        -\begin{pmatrix}
       \hat Q'_{MC}( \Phi_M)v_C\\
       0 \\
       \end{pmatrix}
     \\
     \dot  \Phi_M=&v_C.
\end{split}
\end{align}

A relevant dynamical property is that the SEs in the VCD
admit \emph{invariants of motion}, i.e., functions of the state variables that are constant along the
solutions. Indeed, it is immediate to check that we have
\begin{align}
\begin{split}
\dot \cQ = 0
\end{split}
\end{align}
i.e., $\cQ$ is constant along the solutions to the SEs.
This means that the SEs have $n_C$ invariants of motion given by
the components of function $\cQ$.
As a consequence, the dynamics evolves on the invariant manifolds
\begin{align}
\label{manifold RLCMa}
\cM_{\mathrm{RLCM}a}(Q_0)= \{ (v_C,i_L,\Phi_M) \in \R^{2n_C+n_L}: \cQ=Q_0 \}
\end{align}
where $Q_0 \in \R^{n_C}$ is termed \emph{manifold index}.

\subsection{Subclass RLCM$b$}
\label{sect:subc RLCMb}

Consider now the subclass RLMC$b$ where only $n_{\Phi}<n_C$ capacitors have in parallel a flux-controlled memristor and there is no other memristor. Without loss of generality, we can assume that those are the first $n_{\Phi}$ capacitors, and introduce the partition $\varphi_C=(\varphi_{CM}, \varphi_{CN})$, where $\varphi_{CM}=(\varphi_{Ci})_{i=1,\dots, n_\Phi}$, $\varphi_{CN}=(\varphi_{Ci})_{i=n_\Phi+1,\dots, n_C}$. We have $\Phi_M=\varphi_{CM}+\Phi_{M0}$. Define the state sub-vector $Z=(\varphi_{CN}, q_L)$ and matrices $C_M={\rm diag }(C_{1}, \dots,C_{n\Phi})$, $M_{CL}={\rm diag }(C_{n\Phi+1}, \dots,C_{nC},$ $ L_1, \dots, L_{nL})$.
Then, the SEs of $\mathfrak{N}$ in the FCD are given by
\begin{align}\label{SEs FCD RLCM case NM<NC}
\begin{split}
        \begin{pmatrix}
      C_M \dot \varphi_{CM} \\
      M_{CL} \dot Z \\
    \end{pmatrix}
=& -H
\begin{pmatrix}
       \varphi_{CM} \\
       Z \\
    \end{pmatrix} \\
       &-\begin{pmatrix}
       \hat Q_{MC}( \varphi_{CM}+\Phi_{M0})-
       \hat Q_{MC}(\Phi_{M0})\\
       0 \\
    \end{pmatrix}\\
    &+
     \begin{pmatrix}
      C_M v_{CM0} \\
      M_{CL} z_{0} \\
    \end{pmatrix}
\end{split}
\end{align}
where $z_0=(v_{CN0},i_{L0})$ and
$$
H=\begin{pmatrix}
        H_{11} & H_{12} \\
        H_{21} & H_{22} \\
      \end{pmatrix}
$$
is the partition induced by the state vector $(\varphi_{CM}, Z)$ on matrix $H$.
Since (\ref{SEs FCD RLCM case NM<NC}) has the same form of (\ref{SEs FCD RLCM case 1}), it is possible to derive, under the assumption that submatrix $H_{22}$ is nonsingular, the expressions of invariants (cf.~(\ref{Q(t) case 1})), the form for the SEs in the FCD (cf.\ (\ref{SEs FCD RLCM case 1 xy})), the SEs in the VCD (cf.~(\ref{SEs VCD RLCM case 1})), and the $n_{CM}$-dimensional invariant manifolds (cf.~(\ref{Q(t) case 1}) and (\ref{manifold RLCMa})), by simply replacing in those equations quantities
$$
(\varphi_C, q_L, C, L, H_{\alpha \alpha}, H_{\beta \beta}, H_{\alpha \beta}, -H_{\alpha, \beta}^\top, v_{C0}, i_{L0})
$$
with
 $$
 (\varphi_{CM}, Z, C_M, M, H_{11}, H_{22}, H_{12}, H_{21}, v_{CM0},z_0).
 $$

\subsection{Subclass RLMC$c$}
\label{sect:subc RLCMc}

Finally, consider the case where each capacitor has in parallel a flux-controlled
memristor and each inductor has in series a charge-controlled memristor.
For this subclass, which is termed RLMC$c$,
we have $n_{\Phi}=n_C$ and $n_{Q}=n_L$, while $\Phi_{MC}=\varphi_{CM}+\Phi_{MC0}$
and $Q_{ML}=q_{LM}+Q_{ML0}$.

It is shown in \cite{DiMarco2026BraytonTheor} that the for subclass RLCMc the SEs (\ref{SEsFCDRLCM})
boil down to
\begin{align}\label{SEs FCD RLCM special}
\begin{split}
    \begin{pmatrix}
      C \dot \Phi_{MC}  \\
      L \dot Q_{ML} \\
    \end{pmatrix}
=&
-H
\begin{pmatrix}
       \Phi_{MC} \\
       Q_{ML} \\
    \end{pmatrix}
        -\begin{pmatrix}
       \hat Q_{MC}( \Phi_{MC} )\\
       \hat \Phi_{ML}( Q_{ML})\\
     \end{pmatrix}
    +
\begin{pmatrix}
       Q_0\\
       \Phi_0 \\
    \end{pmatrix}.
\end{split}
\end{align}
Here, we have let
\begin{align}
\begin{split}
\label{Q0 case2}
Q_0=\cQ(0)=& Cv_{C0}+H_{\alpha \alpha} \Phi_{MC0}+H_{\alpha \beta} Q_{ML0}\\
    &+\hat Q_{MC}(\Phi_{MC0})
\end{split}
\end{align}
and
\begin{align}
\begin{split}
\label{Phi0 case2}
\Phi_0=\Psi(0)=& Li_{L0}-H_{\alpha \beta}^\top \Phi_{MC0}+H_{\beta \beta} Q_{ML0}\\
    &+\hat \Phi_{ML}(Q_{ML0})
\end{split}
\end{align}
where $\cQ:\R^{2(n_C+n_L)} \to \R^{n_C}$ and $\Psi:\R^{2(n_C+n_L)} \to \R^{n_C}$
are defined as
\begin{align}
\begin{split}
\label{Q(t) case2}
    \cQ =&\cQ(v_C,i_L,\Phi_{MC},Q_{ML})\\
    =& Cv_C+H_{\alpha \alpha} \Phi_{MC}+H_{\alpha \beta} Q_{ML}\\
    &+\hat Q_{MC}(\Phi_{MC})
\end{split}
\end{align}
and
\begin{align}
\begin{split}
\label{Phi(t) case2}
    \Psi =&\Psi(v_C,i_L,\Phi_{MC},Q_{ML})\\
    =& Li_L-H_{\alpha \beta}^\top \Phi_{MC}+H_{\beta \beta} Q_{ML}\\
    &+\hat \Phi_{ML}(Q_{ML}).
\end{split}
\end{align}
Again, $Q_0$ and $\Phi_0$ are constant terms depending upon the initial conditions $w_0$ for the state variables in the VCD.

By differentiation, the SEs in the VCD are given by
\begin{align}\label{SEs VCD RLCM special}
\begin{split}
      \begin{pmatrix}
      C \dot v_C \\
      L \dot i_L \\
    \end{pmatrix}
=&
-H
\begin{pmatrix}
       v_C \\
       i_L \\
    \end{pmatrix}
        -\begin{pmatrix}
       \hat Q'_{MC}(\Phi_{MC})v_C \\
       \hat \Phi'_{ML}(Q_{ML})i_L \\
     \end{pmatrix}
     \\
     \dot \Phi_{MC}=&v_C\\
     \dot Q_{ML}=&i_L.
\end{split}
\end{align}
Since
$$
\dot \cQ=0; \ \ \dot \Psi=0
$$
it follows that the SEs (\ref{SEs VCD RLCM special}) admit $n_C+n_L$ invariants of motion given by
the components of functions
$\cQ$ and $\Psi$.
Therefore, the dynamics evolves on the invariant manifolds
\begin{align}
\begin{split}
\cM_{\mathrm{RLCM}c}(Q_0,\Phi_0)=& \{ (v_C,q_L,\Phi_{MC},Q_{ML}) \in \R^{2(n_C+n_L)}:\\
\ \cQ=Q_0, \Psi=\Phi_0 \}
\end{split}
\end{align}
where $(Q_0,\Phi_0) \in \R^{n_C+n_L}$ is the manifold index.

\section{Convergence Results for Class RCM}
\label{sect:RCM conv}

The three subclasses defined in Sect.\ \ref{sect:subclasses} are characterized by the simultaneous
presence of capacitors and inductors. Their convergence properties will be studied in detail
Sect.\ \ref{sect:convRLCM}. Here, we address convergence in the special case where only one type of
reactive elements is present, i.e., capacitors. An analogous treatment holds when
the only reactive elements are inductors (details are omitted). The goal is to show is that in this case
convergence is guaranteed for any capacitors and circuit parameters under some mild technical assumptions. This result will be
compared in Sect.\ \ref{sect:disc} with the convergence results for subclasses RLCM$a$, $b$ and $c$, which instead
require, roughly speaking, additional assumptions on the relative weight of
the inductors and capacitors.

Under Assumption\ \ref{assu:case 2 RLCM 1}, matrix $H$ in (\ref{H cas2}) reduces to the symmetric
conductance matrix $G \in \R^{n_C \times n_C}$ given by
\begin{equation}\label{matr G}
    G=\begin{pmatrix}
        G_{\Phi \Phi} & G_{\Phi C} \\
        G_{\Phi C}^\top & G_{C C} \\
      \end{pmatrix}=  A_{CR} (R_\cL)^{-1} A_{CR}^\top
\end{equation}
where $G_{\Phi \Phi} \in \R^{n_\Phi \times n_\Phi}$, $G_{\Phi C}\in \R^{  n_\Phi \times (n_C-n_\Phi)}$ and $G_{C C}\in \R^{(n_C-n_\Phi) \times (n_C-n_\Phi)}$.
On this basis, we obtain the SEs describing $\mathfrak{N}$ in the FCD
\begin{align}\label{SEs RCM FCD}
\begin{split}
    C \dot \varphi_C
= &-G  \varphi_C \\
        &-\begin{pmatrix}
       \hat Q_{MC}( \varphi_{CM}+\Phi_{M0})-
       \hat Q_{MC}(\Phi_{M0})\\
       0 \\
     \end{pmatrix}+C v_{C0}.
     \end{split}
\end{align}

The main result on convergence in the FCD is as follows.

\begin{theorem}
\label{th:conv bounded_RCM}
Suppose that Assumptions\ \ref{assu:memp}-\ref{assu:case 2 RLCM 1b} are satisfied.
Then, any bounded solution to (\ref{SEs RCM FCD}) converges to the
set $E$ of EPs
and it converges to a singleton $\bar \varphi_C \in E$ if the EPs are isolated.
\end{theorem}

\emph{Proof.} Since $G$ is symmetric, the SEs (\ref{SEs RCM FCD}) can be written as the
gradient system $C \dot\varphi_C=-\nabla W( \varphi_C)$, where function $W: \R^{n_C} \to \R$
is given by
\begin{align}
\begin{split}
\label{E}
    W( \varphi_C) =& \frac{1}{2} \varphi_C^\top G  \varphi_C
    + \sum_{i=1}^{n_{\Phi}} \int_0^{ \varphi_{CMi}}(\hat Q_{MCi}(\rho+\varphi_{MC0i})\\
    &-        \hat Q_{MCi}(\varphi_{MC0i}))d\rho- \varphi_C^\top C v_{C0}.
\end{split}
\end{align}
Then, $W$ is strictly decreasing along non-stationary solutions to (\ref{SEs RCM FCD})
and the theorem follows from standard results on convergence of gradient systems \cite{GVK022705937}. \qed

We refrain from explicitly giving conditions ensuring boundedness of solutions to (\ref{SEs RCM FCD}).
Conditions of this kind can be obtained as a special case of those found
in Sect.\ \ref{sect:convRLCM} in the more general situations where also inductors are
present.

Differentiating (\ref{SEs RCM FCD}), we obtain the SEs of $\mathfrak{N}$ in the VCD
\begin{align}\label{SEs VCD RCM}
\begin{split}
    C \dot v_C
=&
-G v_C
-\begin{pmatrix}
       \hat Q'_{MC}(\Phi_{M})v_{CM}\\
       0 \\
     \end{pmatrix}
     \\
     \dot \Phi_{M}=&v_{CM}.
\end{split}
\end{align}
We denote by $\cE=\{(\bar v_C, \bar \varphi_C) \in \R^{2n_C}: \bar v_C=0, \bar \varphi_C \in \R^{n_C} \}$ the set of EPs of (\ref{SEs VCD RCM}).
Note that at any EP the capacitor voltages vanish.



The next result addresses convergence of solutions to the SEs in the VCD.

\begin{theorem}
\label{th:conv VCD RCM}
Suppose that Assumptions\ \ref{assu:memp}-\ref{assu:case 2 RLCM 1b} are satisfied.
Then, any bounded solution to the SEs in the VCD (\ref{SEs VCD RCM})
converges to the set $\cE$ of EPs,
moreover, $v_C$ converges to 0. If, for given initial conditions $v_{C0}$ and
$\Phi_{M0}$,
the EPs of the corresponding SEs (\ref{SEs RCM FCD})
in the FCD are isolated, then the solution to (\ref{SEs VCD RCM}) converges to a
singleton $(0,\bar \varphi_C) \in \cE$.
\end{theorem}

\emph{Proof.} Consider the solution $(v_C,\Phi_M)$ to the SEs (\ref{SEs VCD RCM})
with initial condition $(v_{C0},\Phi_{M0})$ at $t=0$. It can be checked that
$\varphi_C$ is the solution of (\ref{SEs RCM FCD}) with initial condition
$\varphi_C(0)=0$. Since $\Phi_M$ is bounded, and $\varphi_C=\Phi_M-\Phi_{M0}$, also
$\varphi_C$ is bounded and, due to Theorem\ \ref{th:conv bounded_RCM},
$\varphi_C$ converges to the set $E$ of EPs of (\ref{SEs RCM FCD}). Moreover,
it is seen from (\ref{SEs RCM FCD}) that $v_C=C \dot\varphi_C$ converges to 0.
If, in particular, we have isolated EPs in the FCD for
(\ref{SEs RCM FCD}), then $\varphi_C$ converges to a singleton
(cf.\ Theorem\ \ref{th:conv bounded_RCM}).
\qed

%
%
%

\begin{remark}
\label{rem:W and P}
For an RCM circuit, the mixed potential (\ref{mixed n1}) is
given by
\begin{align}\label{mixed RCM}
\begin{split}
    \cP( \varphi_C)=&
    \frac{1}{2}  \varphi_C^\top G  \varphi_C
    +\sum_{i=1}^{n_{\Phi}} \int_0^{ \varphi_{C i}}(\hat Q_{MCi}(\rho+\varphi_{Mi0})\\
    &-
       \hat Q_{MCi}(\varphi_{Mi0}))d\rho
    - \varphi_C^\top C v_{C0}.
\end{split}
\end{align}
Note that
$$
\cP( \varphi_C) = W( \varphi_C)
$$
i.e., the Lyapunov function used in the proof of Theorem\ \ref{th:conv bounded_RCM} coincides with the mixed potential. This is in agreement with the fundamental relation (\ref{SEs FCD RLCM mixed}) and Brayton-Moser theory \cite{brayton1964theoryI}, according to which $\mathfrak{N}$ obeys
$$
C \dot \varphi_C = -\frac{\partial}{\partial  \varphi_C}\cP( \varphi_C)
=-\nabla W( \varphi_C).
$$
\end{remark}

\begin{remark}
From the proof of Theorem\ \ref{th:conv bounded_RCM}, it is seen that convergence of a memristor circuit in RCM follows since its SEs in the FCD are a gradient system. It is worth to stress that this
is intrinsically related to the reciprocity properties discussed in Remark\ \ref{rem:recipro}, which in turn imply that $G$ is a symmetric matrix.
\end{remark}

\begin{remark}
The result in Theorem\ \ref{th:conv bounded_RCM} is analogous to that in \cite[Th.\ 1]{di2025robust}.
It is worth to note that the class RCM here considered, as it can easily be verified, does
not coincide with that considered in the quoted paper. Moreover, as it was
noted in Remark\ \ref{rem:W and P}, the Lyapunov function here used has a clear
physical meaning, since it coincides with the mixed potential, while the Lyapunov function
in \cite{di2025robust} is found via an abstract mathematical argument.
\end{remark}

\section{Convergence Results for Class RLCM}
\label{sect:convRLCM}

We have seen in Sect.\ \ref{sect:RCM conv} that the class RCM of memristor circuits with resistors, capacitors and memristors, but no inductors, has strong convergence properties. This is since, for structural reasons, the SEs describing an RCM memristor circuit in the FCD are the gradient
of a Lyapunov function $W$ coinciding with the Brayton-Moser mixed potential $\cP$
(cf.\ Remark\ \ref{rem:W and P}).
Now, we consider the general class RLCM where capacitors and inductors are simultaneously present and look for conditions that still guarantee convergence.

To study convergence of an RLCM circuit $\mathfrak{N}$,
we will rely on some fundamental results
obtained by Brayton and Moser for nonlinear RLC circuits, in particular Theorems\ 3, 5
in \cite{brayton1964theoryI}. Once more, we can use those results on the basis of
the analogy between an RLCM circuit in the FCD and a nonlinear RLC circuit in the VCD
(Sect.\ \ref{sect:FCAM}).

It is important to point out that, differently from the class RCM, the SEs (\ref{SEsFCDRLCM})
describing a circuit in RLCM are not the gradient of the mixed potential $\cP$
in Sect.\ \ref{sect:mixed}. Indeed, they can be written as follows
\begin{equation}\label{JP}
 -J \dot x
 =\frac{\partial}{\partial x} \cP(x)
\end{equation}
where $x=(\varphi_c,q_L)$ and matrix
$$
J=\begin{pmatrix}
   C & 0 \\
   0 & -L \\
 \end{pmatrix}.
$$
This is not a gradient system, since $J$ is \emph{indefinite} when there are both capacitors and
inductors. The main idea in \cite{brayton1964theoryI} is to describe (\ref{JP}), under suitable
assumptions, by another pair $(J^*,\cP^*)$, i.e.,
\begin{equation}\label{JP_star}
 -J^* \dot x
 =\frac{\partial}{\partial x} \cP^*(x)
\end{equation}
where $J^*$ is such that its symmetric part, $(1/2)(J^*+(J^*)^\top)$ is
 \emph{positive definite}. In this case, (\ref{JP_star}) is again a gradient
system for which convergence can be proved via standard methods.





Recall that, by Assumption\ \ref{assu:case 2 RLCM 1}, any memristor is either in
parallel to a capacitor (if it is flux-controlled) or in series with an inductor
(if it is charge-controlled).
Let $\cG$ be a digraph of $\mathfrak{N}$, $\cT$ a maximal tree of $\cG$ and $\cL$ the corresponding co-tree. Taking into account Assumption\ \ref{assu:case 2 RLCM 1b},
we can choose $\cT$ as follows. Let $\cT'$ be the sub-tree containing all and only the $n_C$
elements $\mathfrak{B}_{CMi}$ and $\mathfrak{B}_{Ci}$
and let $\cL'$ be a sub-co-tree containing all branches of $\cG$ making a loop with branches of $\cT'$, only. Then, $\cL'$ contains only linear resistors. Moreover, $\cL-\cL'$ contains all and only the $n_L$
elements $\mathfrak{B}_{LMi}$ and $\mathfrak{B}_{Li}$, while $\cT-\cT'$ has only linear resistors.



In the next sections, we address separately convergence for the three main subclasses RLCM$a$,
RLCM$b$ and RLCM$c$ (cf.\ Sect.\ \ref{sect:class_RLCM}).
In the analysis, we will make use of these notations.
Given a symmetric matrix $A$, we denote by $\mu(A)$ the least (real) eigenvalue
of $A$. If $A$ is (possibly) non-symmetric, we consider for the matrix the spectral norm
$\| A \|_2= \sqrt{\lambda_\mathrm{max}(A^\top A)}$, where $\lambda_\mathrm{max}(\cdot)$ denotes
the maximum eigenvalue. If $A$ is an $n \times m$ matrix, we let
$\|A\|_\infty=\displaystyle{\max_{i=1,\dots,n} \sum_{j=1,\dots,m} |a_{ij}|}$.

\subsection{Convergence for Subclass RLCM$a$}

Consider first the subclass RLCM$a$ of memristor circuits defined in Sect.\ \ref{sect:subc RLCMa}.
In this case, each capacitor has in parallel a flux-controlled memristor and
there is no other memristor. Recall that each memristor satisfies the standing assumptions
made in Sect.\ \ref{sect:class_RLCM}.


\begin{theorem}
\label{th:RLCMa}
Consider a circuit $\mathfrak{N}$ in the class RLCM$a$ and suppose that
Assumptions\ \ref{assu:memp}-\ref{assu:case 2 RLCM 1b} are
satisfied. If the following conditions hold:

1) we have (cf.\ (\ref{Gmi}))
\begin{equation}\label{Gm}
    G_\mathrm{m}  \doteq \min_{i=1,\dots,n_\Phi} \{ G_{mi} \}  >|\mu  (H_{\alpha \alpha})|;
\end{equation}

2) the symmetric matrix $H_{\beta \beta}$ in (\ref{H cas2}) is positive definite;

3) we have
\begin{equation}\label{norm2 RLCMa}
    \nu \doteq \| L^{1/2} H_{\beta \beta}^{-1} H_{\alpha \beta}^\top C^{-1/2} \|_2<1
\end{equation}

\noindent then, any solution to the SEs (\ref{SEs FCD RLCM case 1}) in the FCD
is bounded and converges to the set $E$ of EPs, moreover,
it converges to a singleton $(\bar \varphi_C,\bar q_L) \in E$ when the EPs
are isolated.
\end{theorem}

\emph{Proof.}
Since there are no charge-controlled memristors in series with inductors and there are $n_\Phi=n_C$ flux-controlled memristors in parallel to capacitors, for any circuit $\mathfrak{N}$ in the class RLCM$a$ the mixed potential (\ref{mixed n1}) can be cast in a semilinear form
\begin{equation}\label{semilinear}
\cP(\varphi_C,q_L)= -\frac{1}{2}q_L^\top \Gamma q_L+ \Upsilon(\varphi_C)+q_L^\top (\gamma \varphi_c-a)
\end{equation}
by letting
$\Gamma=H_{\beta \beta}$, $\gamma=H^\top_{\alpha \beta}$, $a=-Li_{L0}$ and
\begin{align}
\begin{split}
\Upsilon(\varphi_C)=&\frac{1}{2}\varphi^\top_C H_{\alpha \alpha} \varphi_C  +\sum_{i=1}^{n_{C}} \int_0^{ \varphi_{Ci}}(\hat Q_{Mi}(\rho+\varphi_{Mi0})\\
 &-\hat Q_{Mi}(\varphi_{Mi0}))d\rho - \varphi_C^\top C v_{C0}.
\end{split}
\end{align}
This is the case tackled by \cite[Th.\ 3]{brayton1964theoryI}, so we need to prove that conditions 1), 2) and 3) imply that the hypotheses of that theorem are satisfied.
It is easy to check that conditions 2) and 3) are the same required in that theorem for ensuring convergence to the set of EPs, i.e., $\Gamma$ is positive definite and $\| L^{1/2} \Gamma^{-1} \gamma C^{-1/2} \|_2<1$. Then, it remains to show that condition 1) guarantees that  $\Upsilon(\varphi_C)+\|H^\top_{\alpha \beta} \varphi_C \| \to \infty$ as $\|\varphi_C\| \to \infty$. In turn, this ensures boundedness of solutions to the SEs in the FCD.
First, we show that, due to the assumptions on the function $\hat Q_{Mi}(\Phi_{Mi})$ we have, for any $\Phi_{Mi}$
\begin{equation}
\label{eq:min_int}
\int_0^{\Phi_{Mi}} \hat Q_{Mi}(\rho) d\rho \geq \frac{1}{2}G_{mi}(\Phi_{Mi}^2-\tilde{\Phi}_{Mi}^2).
\end{equation}
Consider two cases: first, let us assume that $|\Phi_{Mi}| \leq \tilde{\Phi}_{Mi}$. We have
\begin{align*}
\int_0^{\Phi_{Mi}} & \hat Q_{Mi}(\rho) d\rho = \int_0^{\Phi_{Mi}} \hat Q_{Mi}(\rho) d\rho \\
 &+ \frac{1}{2}G_{mi}\Phi_{Mi}^2-\frac{1}{2}G_{mi}\Phi_{Mi}^2 \\
  \geq & \frac{1}{2}G_{mi}\Phi_{Mi}^2 + \min_{|\Phi_{Mi}|<\tilde{\Phi}_{Mi}}\left( \int_0^\varphi \hat Q_{Mi}(\rho) d\rho \right)\\
 &- \max_{|\Phi_{Mi}|<\tilde{\Phi}_{Mi}} \left(\frac{1}{2}G_{mi}\varphi^2 \right)
 =\frac{1}{2}G_{mi}\Phi_{Mi}^2 - \frac{1}{2}G_{mi}\tilde{\Phi}_{Mi}^2.
\end{align*}
Now, assume $|\Phi_{Mi}| > \tilde{\Phi}_{Mi}$, Due to~(\ref{Gmi}), also in this second case we obtain
\begin{align*}
\int_0^{\Phi_{Mi}} & \hat Q_{Mi}(\rho) d\rho = \int_0^{\frac{|\Phi_{Mi}|}{\Phi_{Mi}}\tilde{\Phi}_{Mi}} \hat Q_{Mi}(\rho) d\rho \\
 &+ \int_{\frac{|\Phi_{Mi}|}{\Phi_{Mi}}\tilde{\Phi}_{Mi}}^{\Phi_{Mi}} \hat Q_{Mi}(\rho) d\rho \\
  \geq & \int_{\tilde{\Phi}_{Mi}}^{|\Phi_{Mi}|} G_{mi}\rho d\rho= \frac{1}{2} G_{mi} \Phi_{Mi}^2- \frac{1}{2} G_{mi} \tilde{\Phi}_{Mi}^2.
\end{align*}
Since~(\ref{eq:min_int}) holds, there exist $\nu_{i0}$ and $\nu_{i1}$ such that
$
\int_0^{ \varphi_{Ci}}(\hat Q_{Mi}(\rho+\varphi_{Mi0})-
       \hat Q_{Mi}(\varphi_{Mi0}))d\rho \ge \frac{1}{2} G_{mi} \varphi_{Ci}^2 + \nu_{i1} \varphi_{Ci}+ \nu_{i0}
$
and, for $\varphi_C \ne 0$,
\begin{align*}
\sum_{i=1}^{n_{C}} & \int_0^{ \varphi_{Ci}}(\hat Q_{Mi}(\rho+\varphi_{Mi0})- \hat Q_{Mi}(\varphi_{Mi0}))d\rho \\
 \geq & \frac{1}{2} \varphi_C^\top \mathrm{diag}(G_{m,1}, \dots,G_{m,nC}) \varphi_C \\
 & + (\nu_{11}, \dots, \nu_{nC1})^\top \varphi_C + \sum_{i=1}^{n_C} \nu_{i0} \\
  \geq &  \|\varphi_C\|^2 \left( G_\mathrm{m} \! + \! (\nu_{11}, \dots, \nu_{nC1})^\top \frac{\varphi_C}{\|\varphi_C\|^2} + \frac{\sum_{i=1}^{n_C} \nu_{i0}}{\|\varphi_C\|^2} \right).
\end{align*}
Taking into account that $\varphi^\top_C H_{\alpha \alpha} \varphi_C \geq \mu (H_{\alpha \alpha}) \|\varphi_C\|^2$, we can write
\begin{align*}
\begin{split}
\Upsilon(\varphi_C)  \ge & \left(\frac{\mu (H_{\alpha \alpha})+G_\mathrm{m} }{2} + (\nu_{11}, \dots, \nu_{nC1})^\top \frac{\varphi_C}{\|\varphi_C\|^2} \right. \\
& + \left. \frac{\sum_{i=1}^{n_C} \nu_{i0}}{\|\varphi_C\|^2} \right) \|\varphi_C\|^2.
\end{split}
\end{align*}
From condition 1) we have $\mu (H_{\alpha \alpha})+G_\mathrm{m} >0$. Therefore, $\Upsilon(\varphi_C)+\|H^\top_{\alpha \beta} \varphi_C \| \geq \Upsilon(\varphi_C) \to \infty$ as $ \|\varphi_C\| \to \infty$, thus completing the proof. \qed

Consider now the SEs (\ref{SEs VCD RLCM case 1}) describing a memristor circuit
$\mathfrak{N}$ in the VCD. The following holds.

\begin{theorem}
\label{th:conv VCD RlCMa}
Suppose that the assumptions of Theorem\ \ref{th:RLCMa}
are satisfied.
Then, any solution to the SEs in the VCD (\ref{SEs VCD RLCM case 1})
is bounded and converges to the set $\cE$ of EPs,
moreover, $v_C$ and $i_L$ converge to 0. If, for given initial conditions $w_0$
in the VCD,
the EPs of the corresponding SEs (\ref{SEs FCD RLCM case 1})
in the FCD are isolated, then the solution to (\ref{SEs VCD RLCM case 1}) with
initial conditions $w_0$ converges to a
singleton $(\bar v_C=0, \bar i_L=0, \bar \varphi_M, \bar Q_M) \in \cE$.
\end{theorem}

\emph{Proof.}
Analogous to that of Theorem\ \ref{th:conv VCD RCM}.
\qed


\begin{remark}
On the basis of the proof of Theorem\ \ref{th:RLCMa}, and Theorem\ 3 in \cite{brayton1964theoryI},
the SEs (\ref{SEs FCD RLCM case 1}) in the FCD describing a memristor circuit in
the class RLCM$a$ can be put in the gradient form (\ref{JP_star}),
where $J^*$ and $P^*$ are explicitly given as
$$
J^*= \begin{pmatrix}
C & - 2 H_{\alpha \beta} H^{-1}_{\beta \beta} L \\
0 & L
\end{pmatrix}
$$
%
%
and
\begin{align*}
\begin{split}
& \cP^*(\varphi_C,q_L) = \cP(\varphi_C,q_L)
  + (i^\top_{L0}L+\varphi^\top_C H_{\alpha \beta} - q^\top_L H_{\beta \beta}) \cdot \\
  & \cdot H^{-1}_{\beta \beta} (L i_{L0}L+ H^\top_{\alpha \beta} \varphi_C  - H_{\beta \beta} q_L).
\end{split}
\end{align*}

\end{remark}

\begin{remark}
\label{rem:condTha}
Let us discuss in more detail the role of Conditions 1)-3) in Theorem\ \ref{th:RLCMa}
in view of the applications.

From the
theorem proof, Conditions\ 1) and 2) are seen to ensure boundedness of solutions,
while Condition\ 3) guarantees convergence of solutions.

Roughly speaking, Condition\ 1) requires that for large $\Phi_{Mi}$ the slope of the memristor characteristics $\hat Q(\Phi_{Mi})$ is positive and it dominates the eigenvalues (in absolute value) of matrix $H_{\alpha \alpha}$.

Consider now Condition\ 2).
Matrix $H_{\beta \beta}=A_{RL}^\top R_{\cT} A_{RL}$ in that condition depends upon the linear resistors belonging to
$\cT -\cT'$ (cf.\ (\ref{H cas2})). If such resistors are positive (i.e., $R_{\cT}$ is a diagonal positive definite matrix), then we can write $H_{\beta \beta}=(R_{\cT}^{-1/2}A_{RL})^\top (R_{\cT}^{-1/2}A_{RL})$, which is positive definite if and only if $\mathrm{rank}(R_{\cT}^{-1/2}A_{RL})=\mathrm{rank}(A_{RL})=n_L$.
This requires in particular that the number of linear resistors in $\cT -\cT'$ is equal to
or larger than the number $n_L$ of inductors.

If Condition\ 2) holds,
i.e., $H_{\beta \beta}$ is positive definite, then condition
3) is certainly satisfied provided the inductances (resp., capacitances)
are sufficiently small (resp., sufficiently large), as it is discussed in
Example\ 1.
\end{remark}

\begin{remark}
\label{rem:multEPsTha}
A relevant role is played by active (i.e., negative) resistors. Indeed, if all
linear resistors are positive, it can be shown that the SEs (\ref{SEs FCD RLCM case 1}) in the FCD have only
one stable EP attracting all solutions. Negative resistors are therefore
necessary for multiple stable EPs.
On the basis of the previous discussion, if linear resistors in $\cT -\cT'$ are
chosen positive to make $H_{\beta \beta}$ positive definite, then we need to place
linear negative resistors in $\cL-\cL'$. Note that the symmetric matrix
$H_{\alpha \alpha}$ depends upon such linear negative resistors and is in general
indefinite (cf.\ (\ref{H cas2})).
\end{remark}

\subsection{Convergence for Subclass RLCM$b$}

Consider now the subclass RLCM$b$ of memristor circuits defined in Sect.\ \ref{sect:subc RLCMb}.
In this case, there are $n_\Phi<n_C$ flux-controlled memristors in parallel to capacitors and
there is no other memristor.

\begin{theorem}
\label{th_RLCMb}
Consider a circuit $\mathfrak{N}$ in the class RLCM$b$ and suppose that
Assumptions\ \ref{assu:memp}-\ref{assu:case 2 RLCM 1b} are
satisfied. Suppose that
the following conditions hold:

1) any active resistor $R_i<0$ is in parallel to a capacitor and also to
a flux-controlled memristor $M_{\Phi i}$, moreover, we have $G_{mi}>|1/R_i|$ (cf.\ (\ref{Gmi}));

2) the combined $n_C \times (n_C+n_L)$ matrix
\begin{equation}\label{combined}
\begin{pmatrix}
  A_{CL}, A_{CR} \\
\end{pmatrix}
\end{equation}
has rank equal to the number $n_C$ of capacitors (cf.\ (\ref{P partition}));

3) the symmetric matrix $H_{\beta \beta}$ in (\ref{H cas2}) is positive definite;

4) we have
\begin{equation}\label{norm2 RLCMabis}
    \| L^{1/2} H_{\beta \beta}^{-1} H_{\alpha \beta}^\top C^{-1/2} \|_2<1.
\end{equation}

\noindent Then, any solution to the SEs (\ref{SEs FCD RLCM case NM<NC}) in the FCD
is bounded and converges to the set $E$ of EPs, moreover,
it converges to a singleton $(\bar \varphi_C,\bar q_L) \in E$ when the EPs
are isolated.
\end{theorem}

\emph{Proof.} The proof is analogous to that of Theorem~\ref{th:RLCMa} and it makes use of \cite[Th.\ 3]{brayton1964theoryI}. Indeed, since there are no charge-controlled memristors and there are $n_\Phi<n_C$ flux-controlled memristors, for any circuit $\mathfrak{N}$ in the class RLCM$b$ the mixed potential (\ref{mixed n1}) can be cast in the same semilinear form~(\ref{semilinear}), with the only difference that $\Upsilon(\cdot)$ now takes the form
\begin{align*}
\begin{split}
\Upsilon(\varphi_C)=&\frac{1}{2}\varphi^\top_C H_{\alpha \alpha} \varphi_C  +\sum_{i=1}^{n_{\Phi}} \int_0^{ \varphi_{CMi}}(\hat Q_{Mi}(\rho+\varphi_{Mi0})\\
       &-
       \hat Q_{Mi}(\varphi_{Mi0}))d\rho - \varphi_C^\top C v_{C0}.
\end{split}
\end{align*}
Conditions 3) and 4) are exactly Conditions 2) and 3) of \cite[Th.\ 3]{brayton1964theoryI}, respectively, of Theorem~\ref{th:RLCMa}. To complete the proof, we have to show that Conditions 1) and 2) guarantee that  $\Upsilon(\varphi_C)+\|H^\top_{\alpha \beta} \varphi_C\| \to \infty$ as $\|\varphi_C\| \to \infty$.

To this end, we first note that the parallel interconnection of a resistor $R_i$ with a flux-controlled memristor is described by the CR
\begin{equation}\label{RM}
    q_i(\varphi_{Mi})=\hat \cQ_{Mi}(\varphi_{Mi}) \doteq \frac{\varphi_{Mi}}{R_i}+ \hat Q_{Mi}(\varphi_{Mi})
\end{equation}
hence, it is equivalent to a flux-controlled memristor with characteristic
$\hat \cQ_{Mi}$.
Consequently, under Condition 1) any (possibly nonlinear) equivalent resistor satisfies
(\ref{Gmi}) with $\cG_{Mi}=G_{Mi}+1/R_i$. This and Condition 2) guarantee that $\Upsilon(\varphi_C)+\|H^\top_{\alpha \beta} \varphi_C\| \to \infty$ as $\|\varphi_C\| \to \infty$
(cf.\ \cite[Sect.\ 19]{brayton1964theoryII}). \qed

\begin{remark}
For a circuit in the class RLCM$b$, we can repeat analogous considerations as in
Remark\ \ref{rem:condTha}, \ref{rem:multEPsTha}.
Moreover, in the VCD,
a result analogous to Theorem\ \ref{th:conv VCD RlCMa}
can be proved (details are omitted).
\end{remark}

\subsection{Convergence for Subclass RLCM$c$}

\label{sect:conv RLCMc}

Consider now class RLMC$c$ defined in Sect.\ \ref{sect:subc RLCMc}, i.e., the case where each capacitor has in parallel a flux-controlled
memristor and each inductor has in series a charge-controlled memristor, hence $n_C=n_\Phi$ and $n_L=n_Q$.

We need some additional notations. Under Assumptions\ \ref{assu:memp}, \ref{assu:memc},
for each flux-controlled memristor we let
\begin{equation}\label{rho}
    \rho_{Mi} \doteq \inf_{\Phi_{Mi} \in \R} \hat Q_{Mi}'(\Phi_{Mi}) \ge 0
\end{equation}
and for each charge-controlled memristor
\begin{equation}\label{gamma}
    \gamma_{Mi} \doteq \inf_{Q_{Mi} \in \R} \hat \Phi_{Mi}'(Q_{Mi})  \ge 0.
\end{equation}

\begin{theorem}
\label{th_RLCMc}
Consider a circuit $\mathfrak{N}$ in the class RLCM$c$ and suppose that
Assumptions\ \ref{assu:memp}-\ref{assu:case 2 RLCM 1b} are
satisfied. Suppose that
$$
\mu_1+\mu_2>0
$$
where we let
\begin{equation}\label{mu1}
\mu_1 \doteq
    \mu(L^{-1/2}(H_{\beta \beta}+{\rm diag}(\rho_{M1},\dots,\rho_{Mn_{\Phi}}))L^{-1/2})
\end{equation}
and
\begin{equation}\label{mu2}
\mu_2 \doteq    \mu(C^{-1/2}(H_{\alpha \alpha}+{\rm diag}(\gamma_{M1},\dots,\gamma_{Mn_{Q}}))C^{-1/2}).
\end{equation}

\noindent Then, any bounded solution to the SEs~(\ref{SEs FCD RLCM special}) in the FCD converges to the set $E$ of EPs,
moreover, it converges to a singleton $(\bar \varphi_C,\bar q_L) \in E$ when the EPs
are isolated.
\end{theorem}

\emph{Proof.}
To show that each bounded solution converges to the set of EPs, we will make use of Theorem 5 in \cite{brayton1964theoryI}.
We first note that we can recast the expression of the mixed potential~(\ref{mixed n1}) in the form
$$
\cP(\varphi_C,q_L)= -\Gamma(q_L)+ \Upsilon(\varphi_C)+q_L^\top \gamma \varphi_c
$$
once we choose
\begin{align*}
\Gamma(q_L)=&\frac{1}{2}q_L^\top H_{\beta \beta} q_L  +\sum_{i=1}^{n_{L}} \int_0^{ q_{Li}}(\hat \Phi_{Mi}(\rho+q_{Mi0})\\
&- \hat \Phi_{Mi}(q_{Mi0}))d\rho - q_L^\top L i_{L0}
\end{align*}
\begin{align}
\Upsilon(\varphi_C)=&\frac{1}{2}\varphi^\top_C H_{\alpha \alpha} \varphi_C  +\sum_{i=1}^{n_{C}} \int_0^{ \varphi_{Ci}}(\hat Q_{Mi}(\rho+\varphi_{Mi0})\\
&-\hat Q_{Mi}(\varphi_{Mi0}))d\rho - \varphi_C^\top C v_{C0}
\end{align}
and
$
\gamma=H_{\alpha \beta}^\top.
$
To use Theorem\ 5, we have to prove that
\begin{equation}
\label{eq:condition-Theo-5-BM}
\mu (L^{-1/2} \nabla^2 \Gamma(q_L) L^{-1/2}) + \mu (C^{-1/2} \nabla^2 \Upsilon(\varphi_C)C^{-1/2}) >0
\end{equation}
for all $q_L$ and all $\varphi_C$, where
$$
\nabla^2 \Gamma(q_L)= \frac{\partial^2 \Gamma(q_L)}{\partial{q_{Li}} \partial{q_{Lj}}},
\ \ \nabla^2 \Upsilon(\varphi_C)= \frac{\partial^2 \Upsilon(\varphi_C)}{\partial{\varphi_{Ci}} \partial{\varphi_{Cj}}}.
$$
Since
\begin{align*}
& \nabla^2 \Gamma(q_L)= H_{\beta \beta}\\
&+ \mathrm{diag}\left(\Phi'_{M1}(q_{L1}+Q_{M10}),\dots,\Phi'_{MnL}(q_{nL}+Q_{MnL0})\right)
\end{align*}
is a symmetric matrix, then for all $q_L$ we have
\begin{align*}
&\mu (L^{-1/2} \nabla^2 \Gamma(q_L) L^{-1/2})\\
 = &\min_{\|x\|=1} \left ( x^\top L^{-1/2} \nabla^2 \Gamma(q_L) L^{-1/2} x \right ) \\
 = &
\min_{\|x\|=1} \Big ( x^\top L^{-1/2} H_{\beta \beta} L^{-1/2} x \\
 &+   x^\top L^{-1/2} \mathrm{diag}(\Phi'_{M1}(Q_{L1}+Q_{M10}), \\
 & \dots,\Phi'_{MnL}(Q_{nL}+Q_{MnL0})) L^{-1/2} x \big)  \\
\geq & \min_{\|x\|=1} \Big ( x^\top L^{-1/2} H_{\beta \beta} L^{-1/2} x+  x^\top L^{-1/2} {\rm diag}(\rho_{M1}, \\
& \dots,\rho_{Mn_{\Phi}}) L^{-1/2} x \Big )\\
= &  \mu(L^{-1/2}(H_{\beta \beta}+{\rm diag}(\rho_{M1},\dots,\rho_{Mn_{\Phi}})L^{-1/2})= \mu_1.
\end{align*}
Using the same argument, we can also prove that
\begin{align*}
& \mu (C^{-1/2} \nabla^2 \Upsilon(\varphi_C)C^{-1/2}) \\
 & \ge  \mu(C^{-1/2}(H_{\alpha \alpha}+{\rm diag}(\gamma_{M1},\dots,\gamma_{Mn_{Q}}))C^{-1/2})=\mu_2.
\end{align*}
Hence, we have shown that $\mu_1+\mu_2>0$ guarantees that~(\ref{eq:condition-Theo-5-BM}) holds true.
In turn this implies that any solution to (\ref{SEs FCD RLCM special}) converges to the set $E$ of EPs, moreover,
it converges to a singleton $(\bar \varphi_C,\bar q_L) \in E$ when the EPs are isolated. \qed

The next result gives a sufficient condition for boundedness of solutions.

\begin{prop}
\label{th_RLCMc-bound}
Consider a circuit $\mathfrak{N}$ in the class RLCM$c$ and suppose that
Assumptions\ \ref{assu:memp}-\ref{assu:case 2 RLCM 1b} are
satisfied. If we have (cf.\ (\ref{Gmi}) and (\ref{Rmi}))
\begin{align}
\label{RLMC_condition1}
\begin{split}
\min_{i=1,\dots,n_C} G_{mi}&  > \| H_{\alpha \alpha} , H_{\alpha \beta} \|_\infty \\
\min_{i=1,\dots,n_L} R_{mi}& > \|- H_{\alpha \beta}^\top  ,  H_{\beta \beta}  \|_\infty
\end{split}
\end{align}
\noindent then any solution to the SEs~(\ref{SEs FCD RLCM special}) in the FCD
is bounded.
\end{prop}

\emph{Proof.}
Since $C$ and $L$ are diagonal positive definite matrices, they do not alter the sign of the right-hand side of~(\ref{SEs FCD RLCM special}). Hence, without loss of generality, we can assume for simplicity that $C_i=1$, $i=1,\dots,n_C$ and $L_j=1$, $j=1,\dots,n_L$.
Let us introduce the state vector $Z=(\Phi_M,Q_M) \in \R^{n_C+n_L}$, and consider function $V(Z)=\|Z\|_\infty=\max_{i=1,\dots,n_C+n_L}|z_i|: \R^{n_C+n_L} \to \R$, which is convex, and radially unbounded in $\R^{n_C+n_L}$. Let $Z(t), t \in [0, T]$, where $0<T<+\infty$, be a solution to~(\ref{SEs FCD RLCM special}). Let $k(t) \in \{1,\dots, n_C+n_L\}$ be such that $|Z_{k(t)}(t)|=\|Z(t)\|_\infty$ for $t \geq 0$.
It can be shown that, as long as
\begin{equation}
\label{eq:conditionZ}
|Z_{k(t)}(t)| >\max \left\{ \max_{i=1,\dots,n_C}\{\tilde \Phi_{Mi} \} \!, \! \max_{i=1,\dots, n_L} \{ \tilde Q_{Mi}\} \right\}
\end{equation}
we have (see Appendix\ A): if $k(t) \in {1,\dots,n_C}$, then
\begin{align}
\label{eq:min-dotVC}
\begin{split}
\dot V(Z(t))\leq & - \| Z_{k(t)}(t)\|_\infty \\
& \cdot  \Bigg (G_{Mk(t)} - \frac{ \| Q_0\|_\infty}{ \|Z_{k(t)}(t)\|_\infty} -
\| H_{\alpha \alpha}, H_{\alpha \beta} \|_\infty  \Bigg )
\end{split}
\end{align}
whereas, if $k(t) \in {n_C+1,\dots,n_M}$, then
\begin{align}
\label{eq:min-dotVL}
\begin{split}
\dot V(Z(t))\leq &  - \| Z_{k(t)}(t)\|_\infty \\
& \cdot    \Bigg( R_{Mk(t)-n_C} \! - \! \frac{ \| \Phi_0\|_\infty}{\|Z_{k(t)}(t)\|_\infty} \! \\
& - \! \| -H_{\alpha \beta}^\top, H_{\beta \beta} \|_\infty  \Bigg ).
\end{split}
\end{align}

Now, choose $\varepsilon= \min \{ \varepsilon_1, \varepsilon_2 \}$, where
\begin{align*}
\varepsilon_1=&\frac{1}{2}\left(\min_{i=1,\dots,n_C} \{ G_{mi} \} -\| H_{\alpha \alpha},  H_{\alpha \beta} \|_\infty  \right) >0\\
\varepsilon_2=& \frac{1}{2}\left( \min_{j=1,\dots,n_L} \{R_{mj} \} - \|-H_{\alpha \beta}^\top, H_{\beta \beta}  \|_\infty \right) >0
\end{align*}
and pick
\begin{align*}
R(\varepsilon)=&\max \left\{\max_{i=1,\dots,n_C}\{ \tilde \Phi_{Mi} \}, \max_{i=1,\dots, n_L} \{ \tilde Q_{Mi}\} , \right. \\
& \left. \frac{ \|Q_0\|_\infty}{\varepsilon},
\frac{\|\Phi_0\|_\infty}{\varepsilon} \right  \}.
\end{align*}
Then, for any $Z(t)$ such that $\|Z(t)\|_\infty>R(\varepsilon)$, due to~(\ref{eq:min-dotVC}) and
(\ref{eq:min-dotVL}), we have $\dot V(Z(t)) < -\varepsilon/2<0$.
As a consequence, the set $B_{R(\varepsilon)} = \{Z \in \R^{n_M}: |Z| \leq R(\varepsilon)\}$ is positively invariant and globally attractive for the solutions to~(\ref{SEs FCD RLCM special}). Therefore, these are bounded and defined for $t \geq 0$. \qed

\begin{remark} Let us discuss in more detail conditions of Theorem\ \ref{th_RLCMc}
and Proposition\ \ref{th_RLCMc-bound}.
While conditions in the proposition guarantee boundedness of solutions, those in the theorem guarantee convergence of solutions.

Considering Proposition\ \ref{th_RLCMc-bound},
we can make analogous considerations as in Remark\ \ref{rem:condTha}.
To ensure conditions of Theorem\  \ref{th_RLCMc} hold, it is necessary that we have either
$\mu_1>0$ or $\mu_2>0$. To this end, a possibility is to proceed as follows.
Suppose $H_{\beta \beta}$ is positive definite (cf.\ Remark\ \ref{rem:condTha}).
Then, since ${\rm diag}(\rho_{M1},\dots,\rho_{Mn_{Mf}})$ has
non-negative diagonal entries, we have
$$
\mu(H_{\beta \beta}+{\rm diag}(\rho_{M1},\dots,\rho_{Mn_{Mf}})).
\ge \mu (H_{\beta \beta})
$$
Therefore, we obtain
$$
\mu_1=\mu(L^{1/2}(H_{\beta \beta}+{\rm diag}(\rho_{M1},\dots,\rho_{Mn_{Mf}}))L^{1/2})>0.
$$
If the capacitors are fixed, it follows that such condition holds for
sufficiently small inductors, see the examples in the next section. Dual considerations hold if
$H_{\alpha \alpha}$ is positive definite.
\end{remark}

\begin{remark}
Concerning the existence of multiple EPs, there hold considerations analogous to
those in Remark\ \ref{rem:multEPsTha}, in particular, we should place
negative resistors in $\cL-\cL'$. Moreover, when a negative resistor $R_i$ is in
parallel to a memristor, we need to have $|R_i| > \rho_{Mi}$ (cf.\ (\ref{rho})).
See Example\ 2 for an illustration.
\end{remark}

\section{Discussion}
\label{sect:disc}

Some general remarks are in order to discuss the obtained convergence results.

\begin{remark}[Comparison with existing results]
To the authors knowledge, the obtained results are the only existing ones pertaining to
nonlinear circuits containing all four basic circuit elements. As such, they extend
previous results for nonlinear circuits where memristors are not considered \cite{brayton1964theoryII}
and nonlinear circuits where inductors are not included
\cite{di2025robust}. They are also an extension of convergence results for NNs using ideal memristors
\cite{DFP16,di2017memristor,DiMarco2020LQprog,10144925,DiMarco2022rectifying,di2022convergence}
and those where use is made of memristors modeled as
switching devices, where once more inductors are not considered, see, e.g.,
\cite{WZ12,GWY2013,YCY14,ZSYS13,WWZ12}, and references therein.
\end{remark}

\begin{remark}[Convergence is parameter dependent]
Let us compare the convergence result for RCM circuits (Theorem\ \ref{th:conv bounded_RCM})
with those for RLCM circuits (Theorems\ \ref{th:RLCMa}, \ref{th_RLCMb}, \ref{th_RLCMc}). While
in the former case convergence
is guaranteed for any circuit parameters, in the latter case we need a crucial additional
assumption, i.e., we need a balance between capacitors and inductors (cf., e.g., Condition 3) in Theorem\ \ref{th:RLCMa}).
This means that convergence is parameter dependent for RLCM circuits.
\end{remark}

\begin{remark}[Convergence is robust]
It is stressed that Theorems\ \ref{th:RLCMa}, \ref{th_RLCMb}, \ref{th_RLCMc} give conditions for convergence which
are robust with respect to circuit parameter variations. Consider for example Theorem\ \ref{th:RLCMa}. Conditions 1) and 3) in that
theorem are defined by strict inequalities, while matrix $H_\mathrm{\beta \beta}$ in condition 2)
is for structural reasons symmetric (cf.\ Remark\ \ref{rem:recipro})
and its positive definiteness property is
robust with respect to small perturbations. Analogous considerations hold for the other two theorems.
\end{remark}

\begin{figure}[H]
\begin{center}
\begin{subfigure}{.9 \columnwidth}
  \includegraphics[width=\linewidth]{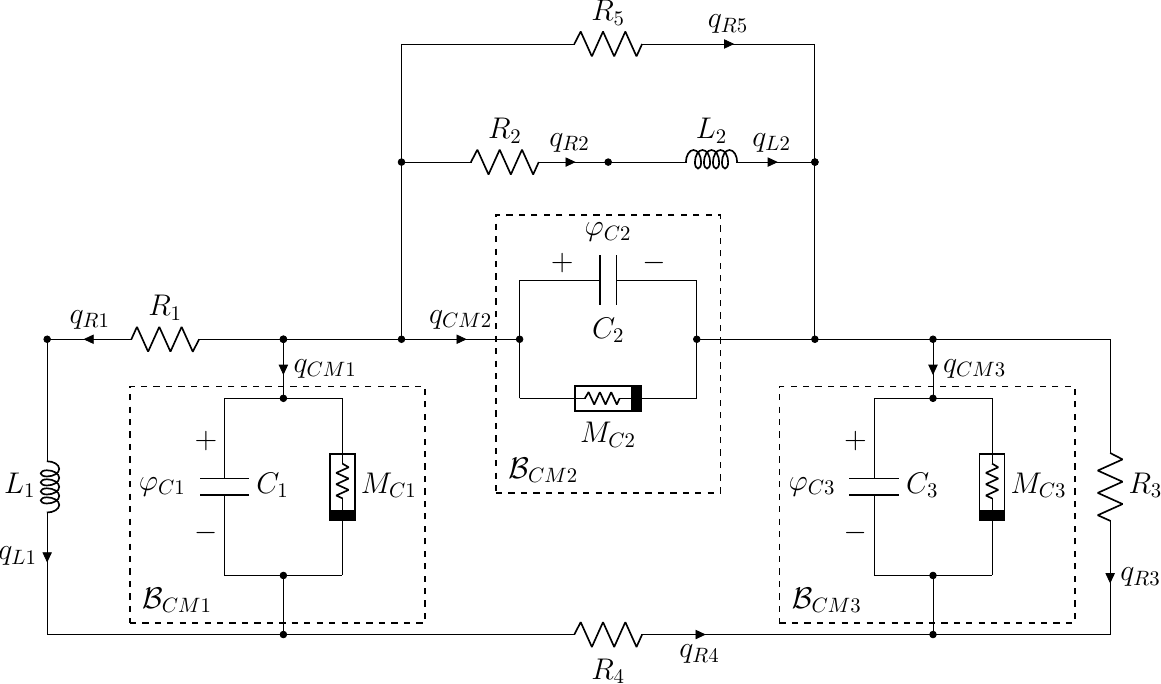}
  \caption{\small \small }
\end{subfigure}\\
\begin{subfigure}{0.57\columnwidth}
  \includegraphics[width=\linewidth]{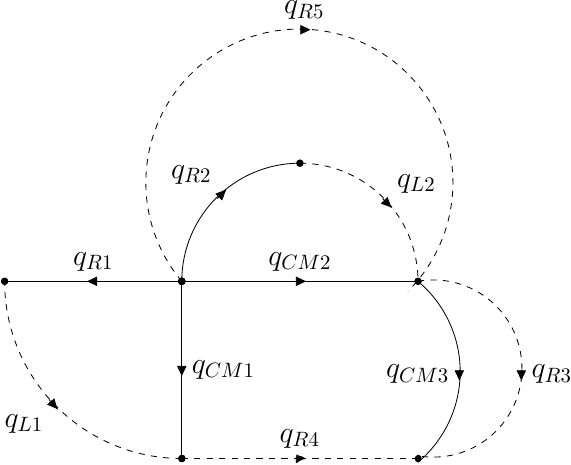}
  \caption{\small \small }
\end{subfigure} \\
\caption{\small (a) Memristor circuit $\mathfrak{N}$ in the class RLCM$a$
in Example\ 1.
(b) Associated digraph where tree branches are indicated by solid lines and
co-tree branches by dashed lines.
}
\label{fig:fluxcirc1}
\end{center}
\end{figure}

\begin{remark}[Convergence is manifold independent]
\label{rem:cond_mani}
It is stressed that the conditions for convergence in the theorems
are manifold independent, i.e., if they are satisfied for a given manifold,
then they are satisfied for any manifold.
\end{remark}

\begin{remark}[Active resistors and multiple stable EPs] We have seen that in any case the active (i.e., negative) resistors play
a crucial role to guarantee the existence of multiple stable EPs. This property is important if we wish to
use RLCM circuits to implement a content addressable memory (CAM) or to solve other signal processing tasks in real time.
\end{remark}

\section{Simulation examples}
\label{sect:ex}
\subsection{Example\ 1}

Consider the memristor circuit $\mathfrak{N}$ in the class RLCM$a$ shown in Fig.\ \ref{fig:fluxcirc1}(a)
(cf.\ Sect.\ \ref{sect:subc RLCMa}).
There are 3 capacitors, 2 inductors and 5 linear resistors, moreover, any capacitor has in parallel a flux-controlled memristor and no other
memristor is present. Choose for instance $\cT=\{ \cB_{CM1},\cB_{CM2},\cB_{CM3},R_1,R_2 \}$, $\cL=\{ \cB_{L1}, \cB_{L2}, R_3,R_4,R_5 \}$, $\cT'=\{ \cB_{CM1},\cB_{CM2},\cB_{CM3} \}$,
$\cL'=\{ R_3,R_4,R_5 \}$, $\cT-\cT'=\{ R_1,R_2 \}$ and $\cL-\cL'=\{ \cB_{L1},\cB_{L2} \}$
(Fig.\ \ref{fig:fluxcirc1}(b)).
It can be checked that Assumptions\ \ref{assu:case 2 RLCM 1}, \ref{assu:case 2 RLCM 1b} are satisfied.
The topological matrix $A$ in (\ref{P partition}) amounts to
\begin{equation}\label{Pex1}
    A=\begin{pmatrix}
        A_{CL} & A_{CR} \\
        A_{RL} & A_{RR} \\
      \end{pmatrix}=
      \left(
      \begin{array}{c c | c c c}
        -1 & 0 & 0 & 1 & 0 \\
        0 & -1 & 0 & -1 & -1 \\
        0 & 0 & -1 & -1 & 0 \\
        \hline
        1 & 0 & 0 & 0 & 0 \\
        0 & 1 & 0 & 0 & 0 \\
      \end{array}
      \right).
\end{equation}

Linear resistors belonging to $\cT-\cT'$ are chosen as $R_1=R_2=0.5$ Ohm, while
resistors in $\cL'$ are $R_3=-0.2$ Ohm, $R_4=5$ Ohm and $R_5=-0.1$ Ohm (normalized values).
On this basis, $R_\cT=\mathrm{diag}(0.5,0.5)$, $R_\cL=\mathrm{diag}(-0.2,5,-0.1)$ and matrix $H$ in (\ref{H cas2}) is given by
\begin{equation}\label{Hex2}
    H \! \! = \! \! \begin{pmatrix}
        H_{\alpha \alpha} & H_{\alpha \beta} \\
        -H_{\alpha \beta}^\top & H_{\beta \beta} \\
      \end{pmatrix} \! \! = \! \!
      \left (
      \begin{array}{c c c | c c}
        0.2 & -0.2 & -0.2 & -1 & 0 \\
        -0.2 & -9.8 & 0.2 & 0 & -1 \\
        -0.2 & 0.2 & -4.8 & 0 & 0 \\
        \hline
        1 & 0 & 0 & 0.5 & 0 \\
        0 & 1 & 0 & 0 & 0.5 \\
      \end{array}
      \right ) \! .
\end{equation}

\begin{figure}[t]
  \centering
\includegraphics[width=0.57\linewidth]{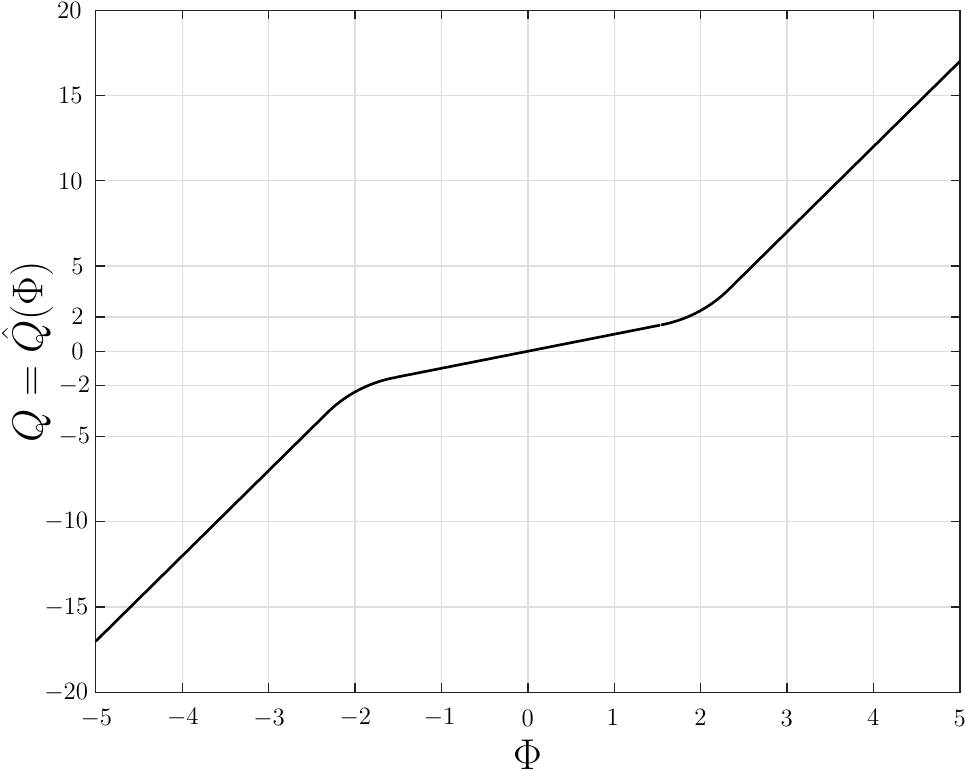}
\caption{\small Smooth $C^1$ approximation of the piecewise linear memristor characteristic $Q=\hat Q(\Phi)
=10 \Phi -(9/2)
(|\Phi+2|-|\Phi-2|)$ used in the examples.}
\label{fig:eleFCD}
\end{figure}

The capacitor values have been set to $C_1=0.1$ F, $C_2=0.25$ F and $C_3=0.2$ F,
while the inductors are $L_1=k_L \cdot 1.5$ H and $L_2=k_L$ H, where $k_L>0$
is a parameter which is varied in the simulations.
All flux-controlled memristors have the same characteristic given by a $C^1$ approximation
of the piecewise-linear function $Q=\hat Q(\Phi)=10 \Phi -(9/2)
(|\Phi+2|-|\Phi-2|)$ (Fig.\ \ref{fig:eleFCD}). In
particular, we have $G_{mi}=G_\mathrm{\mathrm{m}}=10$ Ohm$^{-1}$, $i=1,2,3$
(cf.\ (\ref{Gmi})). All elements have normalized values. We have that
Assumptions\ \ref{assu:case 2 RLCM 1}, \ref{assu:case 2 RLCM 1b} are satisfied.

Since linear resistors belonging to $\cT'$ are positive and
${\rm rank}(A_{CL}\ \  A_{CR})=3=n_C$, then $H_{\beta \beta}$ is positive
definite, as it can be checked from (\ref{Hex2}). Hence, Condition\ 2)
in Theorem\ \ref{th:RLCMa} is satisfied for any $k_L$. Instead,
due to the negative resistors in $\cL'$, $H_{\alpha \alpha}$ is indefinite.
However, $\lambda_\mathrm{min}(H_{\alpha \alpha})=-9.8$,
hence $G_\mathrm{m} =10>|\lambda_\mathrm{min}(H_{\alpha \alpha})|=9.8$.
Consequently,
Condition\ 1) in Theorem\ \ref{th:RLCMa} ensuring boundedness of solutions is
satisfied for any $k_L$.
Finally, it can be checked that Condition\ 3) in Theorem\ \ref{th:RLCMa}
is satisfied for small inductors, namely, when $k_L<k_{L\mathrm{max}}=0.775$,
while this condition fails when $k_L>k_{L\mathrm{max}}$.
Since $H_{\beta \beta}$ is positive definite and, hence, non-singular,
the circuit admit invariants of motion and invariant manifolds as discussed in
Sect.\ \ref{sect:subc RLCMa}.
It is worth to recall that Conditions\ 1), 2) and 3) in Theorem\ \ref{th:RLCMa} are
manifold independent (cf.\ Remark\ \ref{rem:cond_mani}). We also note that,
as it is verified in the next simulations, $k_L$ does not affect
number of EPs of SEs (\ref{SEs FCD RLCM case 1 xy}) while
it affects their stability properties.
The manifold index $Q_0$ affects both the number and stability of EPs
(cf.\ (\ref{Q0 case 1})).

The SEs describing the circuit in the FCD are given in (\ref{SEs FCD RLCM case 1 xy}),
where $Q_0 \in \R^3$ is the manifold index.
In a first experiment we have let $Q_0=0$, i.e., we have considered the
dynamics on the invariant manifold $\cM(Q_0=0)$ and varied parameter $k_L$.
Via a suitable MATLAB program, we have found that the number of EPs of
the SEs (\ref{SEs FCD RLCM case 1 xy}) is equal to 9.

a) Let $k_L=0.01<k_{L\mathrm{max}}$. In this case, since we have $\nu=0.77<1$,
also Condition\ 3) for convergence
in Theorem\ \ref{th:RLCMa} is satisfied
(cf.\ (\ref{norm2 RLCMa})). For this value of $k_L$ there are 4 stable
EPs and 5 unstable EPs. MATLAB simulations show that
any solution is bounded and converges to an EP in accordance with
Theorem\ \ref{th:RLCMa}. Specifically, Fig.\ \ref{fig:fluxcirc1f}
(resp., Fig.\ \ref{fig:chargecirc1c}) shows the capacitor
fluxes (resp., inductor charges), for 4 solutions of
(\ref{SEs FCD RLCM case 1 xy}) converging to different
stable EPs. Also shown are the capacitor voltages (Fig.\ \ref{fig:voltagescirc1v})
and inductor currents (Fig.\ \ref{fig:currentscirc1i}), which  tend to 0 in accordance
with Theorem\ \ref{th:conv VCD RlCMa}. Figure\ \ref{fig:statespacecirc1} shows the behavior of a number
of convergent solutions starting at different initial conditions in the phase space
in the VCD.

b) Let $k_L=0.1>k_{L\mathrm{max}}$. Condition\ 3) in Theorem\ \ref{th:RLCMa} is not satisfied
since $\nu=2.45>1$. We still have 4 stable EPs. From simulations it has been
verified that also in this case solutions appear to be convergent. This is not a contradiction
since the conditions in Theorem\ \ref{th:RLCMa} are only sufficient for convergence.

c) Finally, let $k_L=1>k_{L\mathrm{max}}$. Condition\ 3) in Theorem\ \ref{th:RLCMa} is not satisfied
since $\nu=7.75>1$. There are still 4 stable EPs. In this case, in addition to convergent
solutions there are also solutions that display nonvanishing oscillations
(cf.\ Fig.\ \ref{fig:circ1oscill}). This means that, for the considered circuit, there is coexistence
of stable EPs and a stable limit cycle.

In a final experiment, we have changed the manifold by letting
$Q_0=\zeta(0.8, -0.5, 0.2)$, where $\zeta$ ia a parameter. Via a MATLAB program,
we have found the following. When $\zeta=5$ the
circuit still has 4 stable EPs and 5 unstable EPs; when
$\zeta=6$, there are
3 stable EPs and 4 unstable EPs; when $\zeta=7$, we have
2 stable EPs and three unstable EPs;
when $\zeta=8$, there are
2 stable EPs and one unstable EP and, finally, when
$\zeta=44$ there is only one stable EP. This shows that, as
expected, the number of stable and unstable EPs is largely
dependent upon the manifold.

%
%
%
%
%
%
%

%

\begin{figure}[H]
\begin{center}
\begin{subfigure}{1.\columnwidth}
  \includegraphics[width=\linewidth]{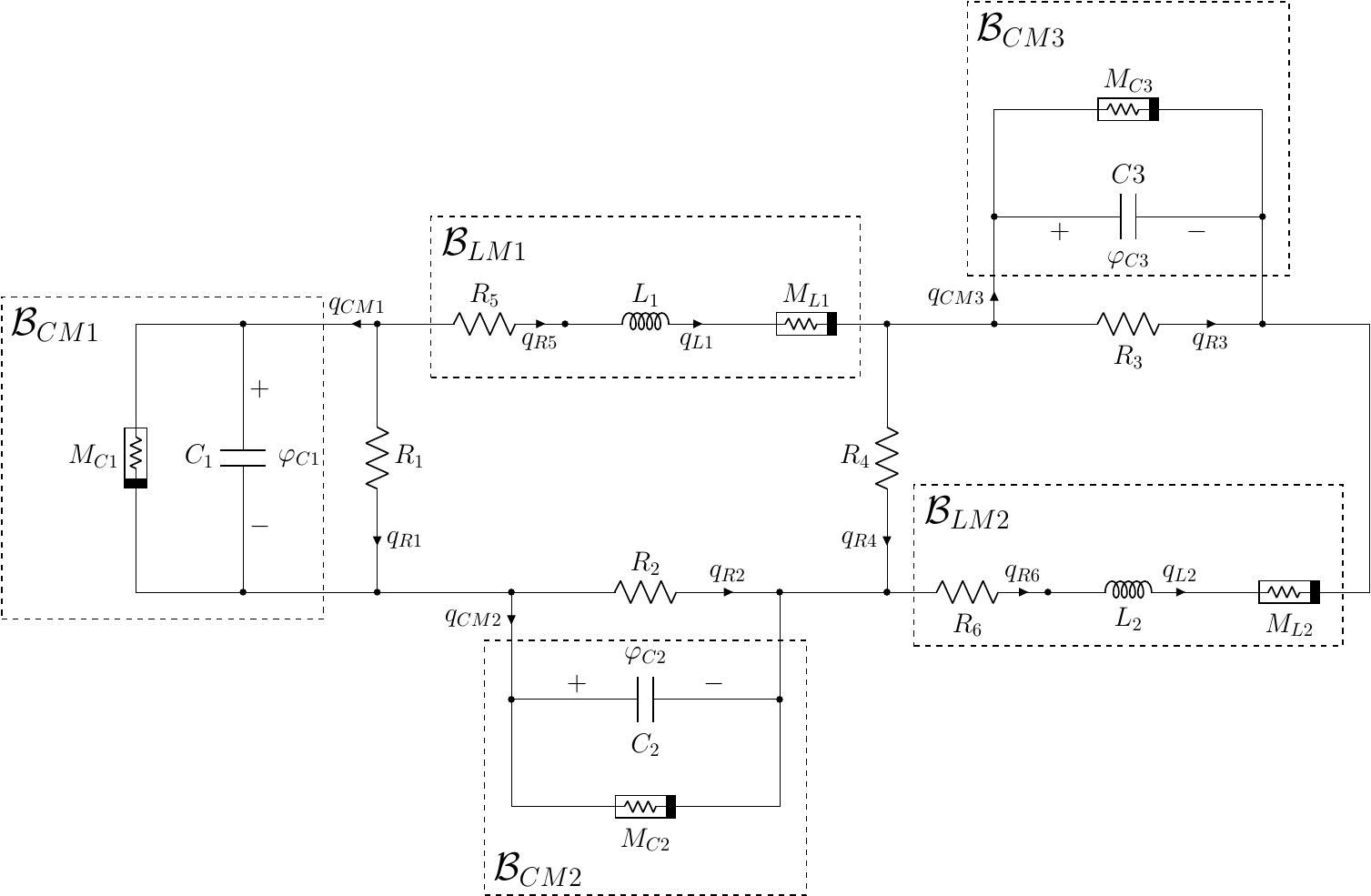}
  \caption{\small \small }
\end{subfigure}\\
\begin{subfigure}{0.7\columnwidth}
  \includegraphics[width=\linewidth]{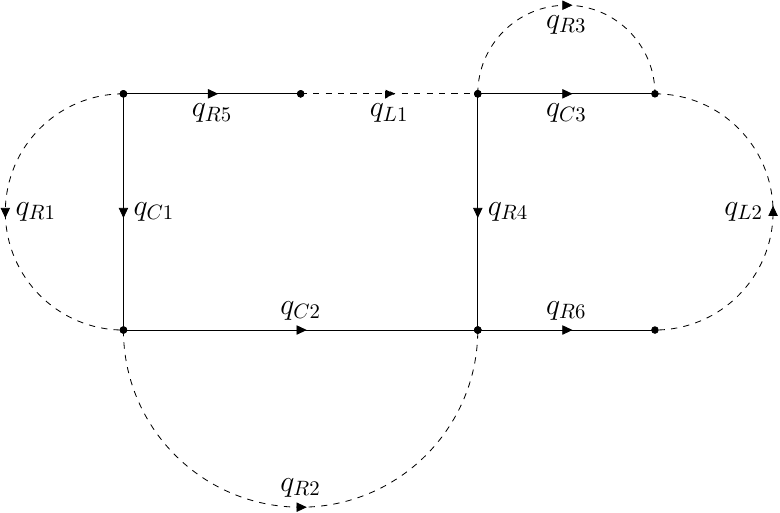}
  \caption{\small \small }
\end{subfigure} \\
\caption{\small (a) Memristor circuit $\mathfrak{N}$ in the class RLCM$c$ in
Example\ 2.
(b) Associated digraph where tree branches are indicated by solid lines and
co-tree branches by dashed lines.
}
\label{fig:fluxcirc2}
\end{center}
\end{figure}


\begin{figure}[t]
\begin{center}
\begin{subfigure}{0.49\columnwidth}
  \includegraphics[width=\linewidth]{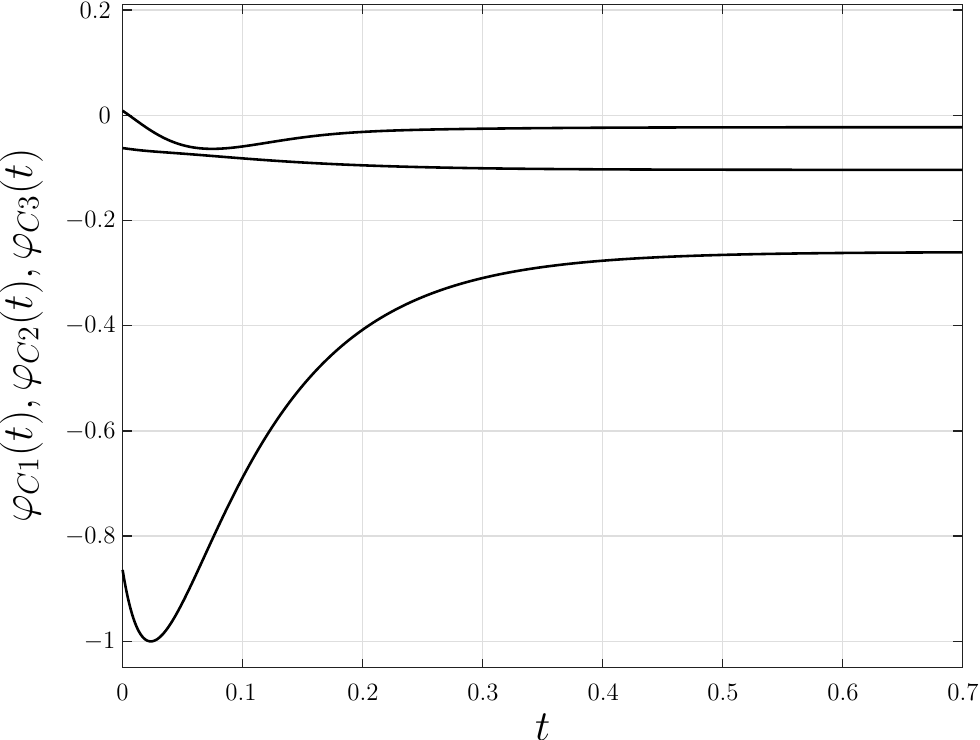}
  \caption{\small \small }
\end{subfigure}
\begin{subfigure}{0.49\columnwidth}
  \includegraphics[width=\linewidth]{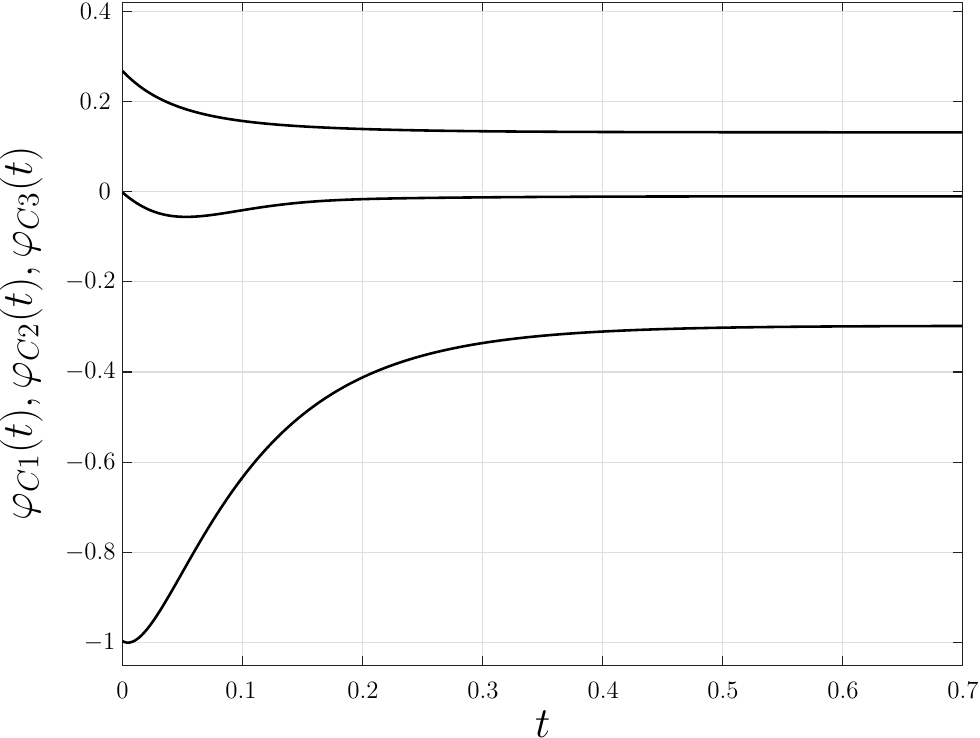}
  \caption{\small \small }
\end{subfigure} \\
\begin{subfigure}{0.49\columnwidth}
  \includegraphics[width=\linewidth]{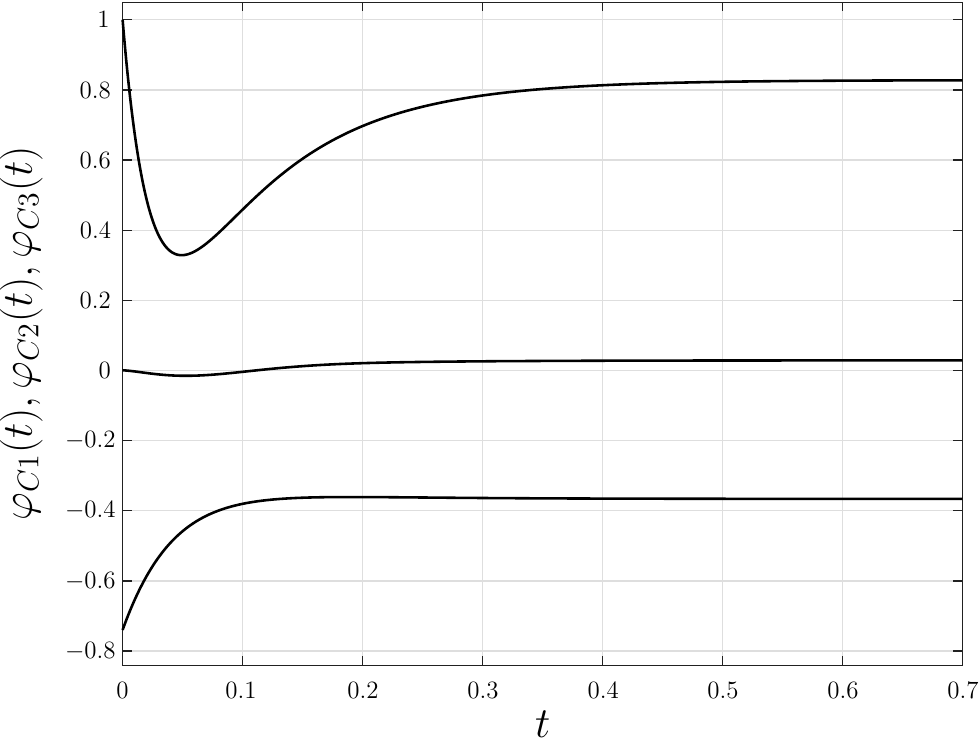}
  \caption{\small \small }
\end{subfigure}
\begin{subfigure}{0.49\columnwidth}
  \includegraphics[width=\linewidth]{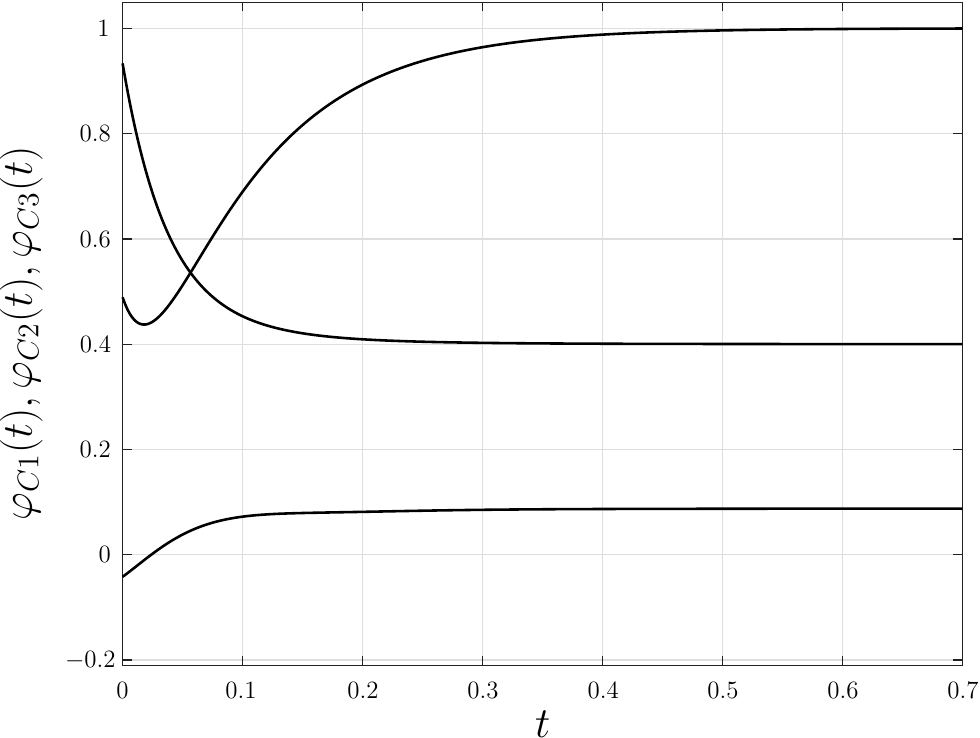}
  \caption{\small \small }
\end{subfigure}\\
\caption{\small Capacitor fluxes $\varphi_{Ci}$, $i=1,2,3$, for four solutions
of the memristor circuit in Example\ 1. Case $k_L=0.01$. Horizontal axis: time. Vertical axis: fluxes in
normalized values.}
\label{fig:fluxcirc1f}
\end{center}
\end{figure}

\begin{figure}[t]
\begin{center}
\begin{subfigure}{0.49\columnwidth}
  \includegraphics[width=\linewidth]{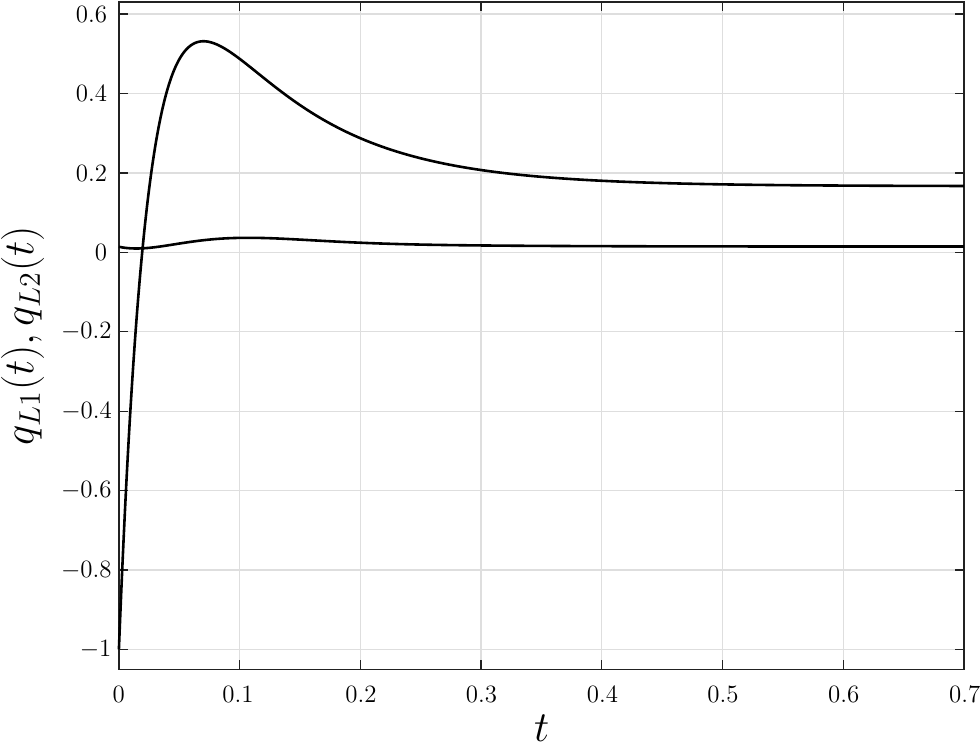}
  \caption{\small \small }
\end{subfigure}
\begin{subfigure}{0.49\columnwidth}
  \includegraphics[width=\linewidth]{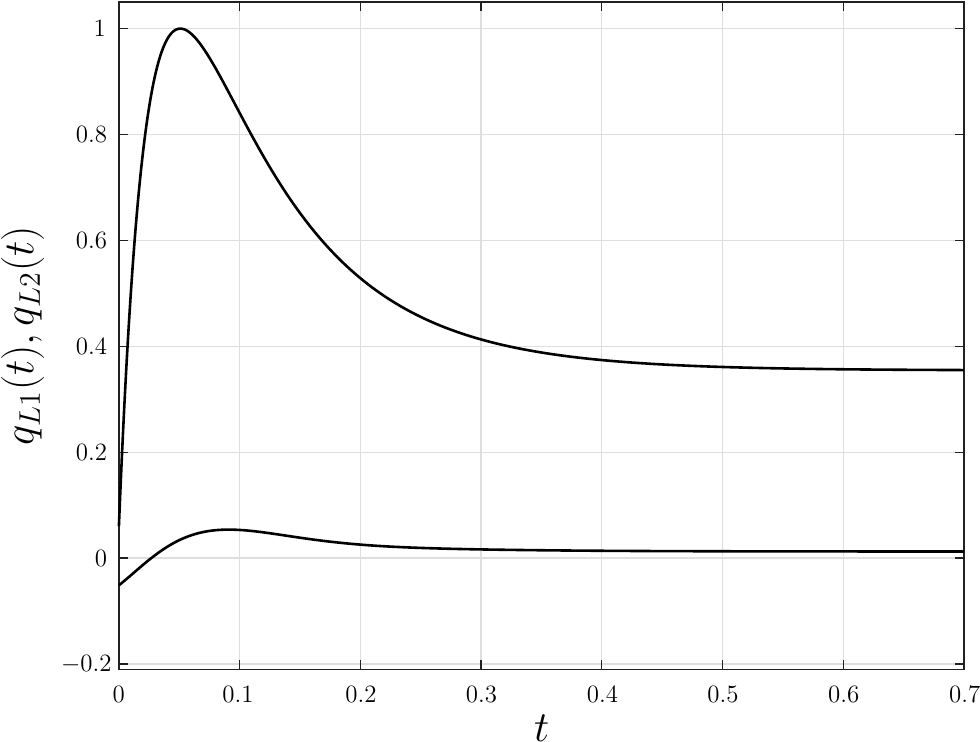}
  \caption{\small \small }
\end{subfigure} \\
\begin{subfigure}{0.49\columnwidth}
  \includegraphics[width=\linewidth]{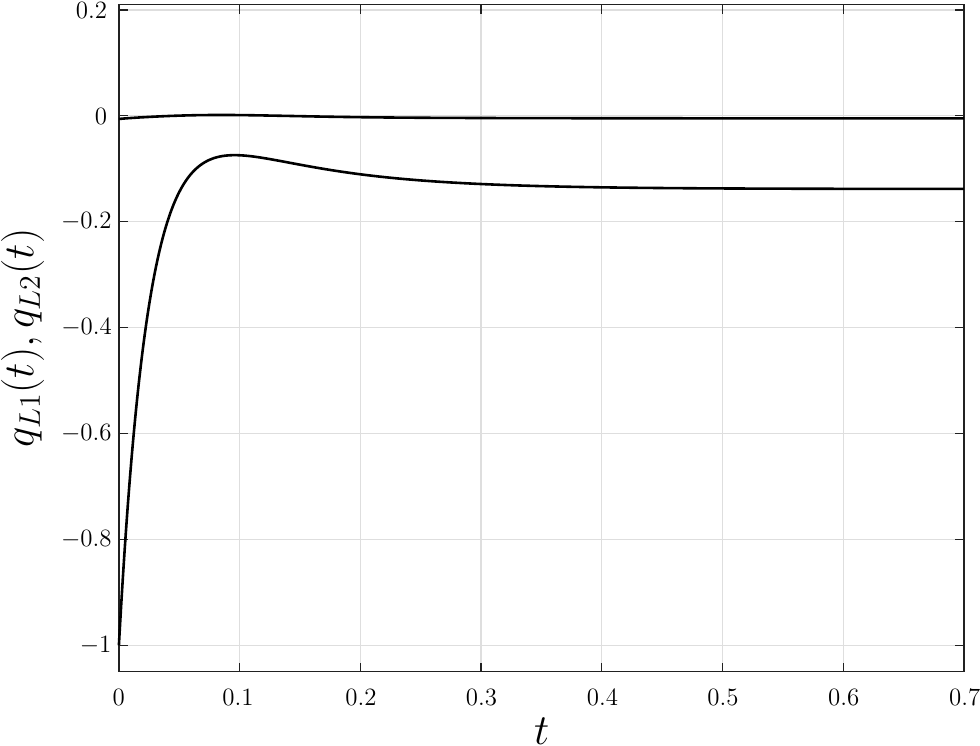}
  \caption{\small \small }
\end{subfigure}
\begin{subfigure}{0.49\columnwidth}
  \includegraphics[width=\linewidth]{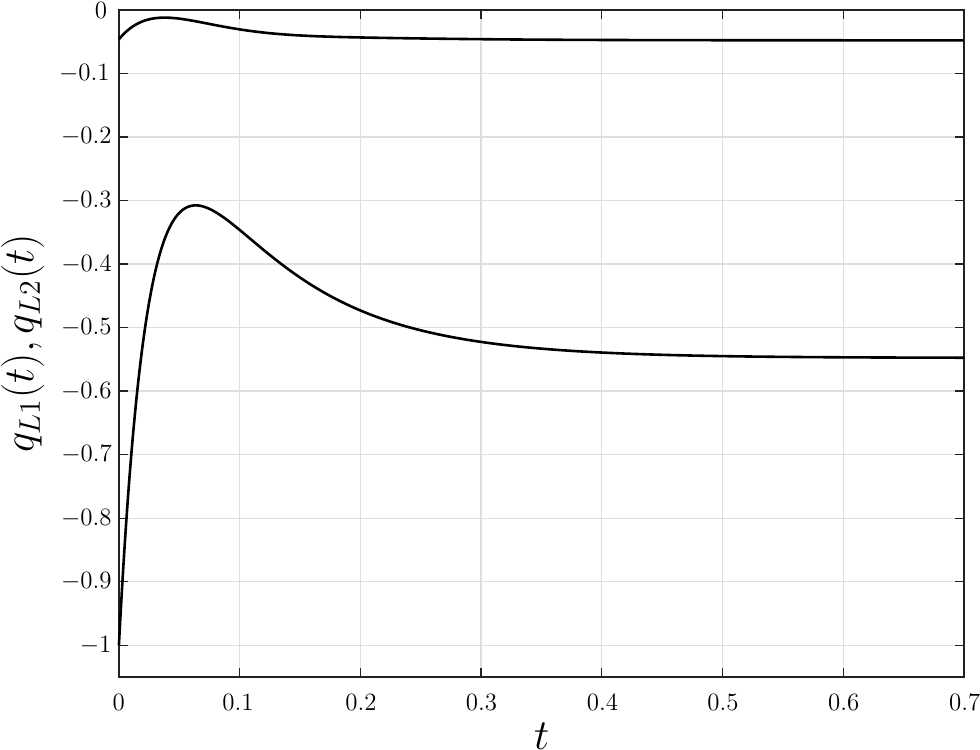}
  \caption{\small \small }
\end{subfigure}\\
\caption{\small Inductor charges $q_{Li}$, $i=1,2$, for four solutions
of the memristor circuit in Example\ 1. Case $k_L=0.01$. Horizontal axis: time. Vertical axis:
inductor charges in normalized values.}
\label{fig:chargecirc1c}
\end{center}
\end{figure}

\begin{figure}[t]
\begin{center}
\begin{subfigure}{0.49\columnwidth}
  \includegraphics[width=\linewidth]{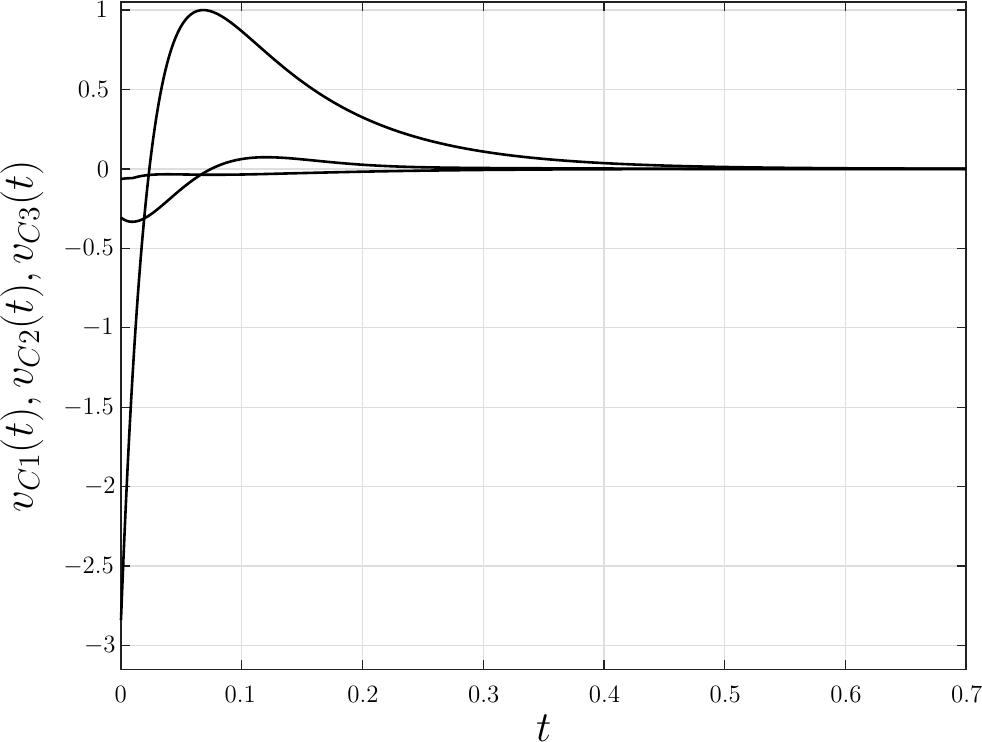}
  \caption{\small \small }
\end{subfigure}
\begin{subfigure}{0.49\columnwidth}
  \includegraphics[width=\linewidth]{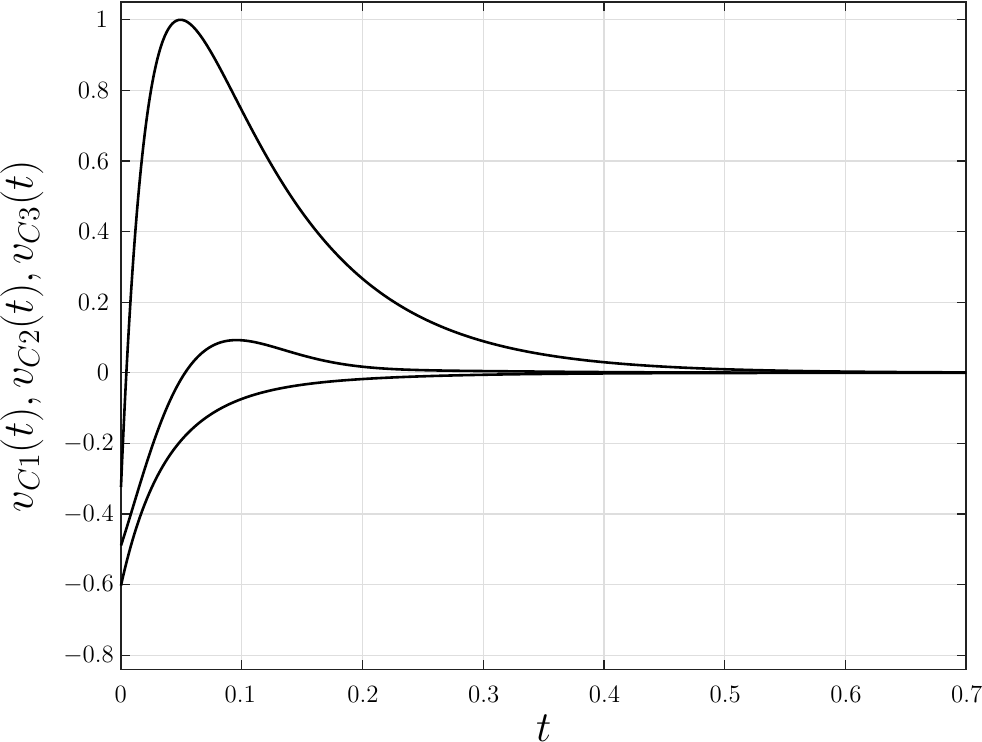}
  \caption{\small \small }
\end{subfigure} \\
\begin{subfigure}{0.49\columnwidth}
  \includegraphics[width=\linewidth]{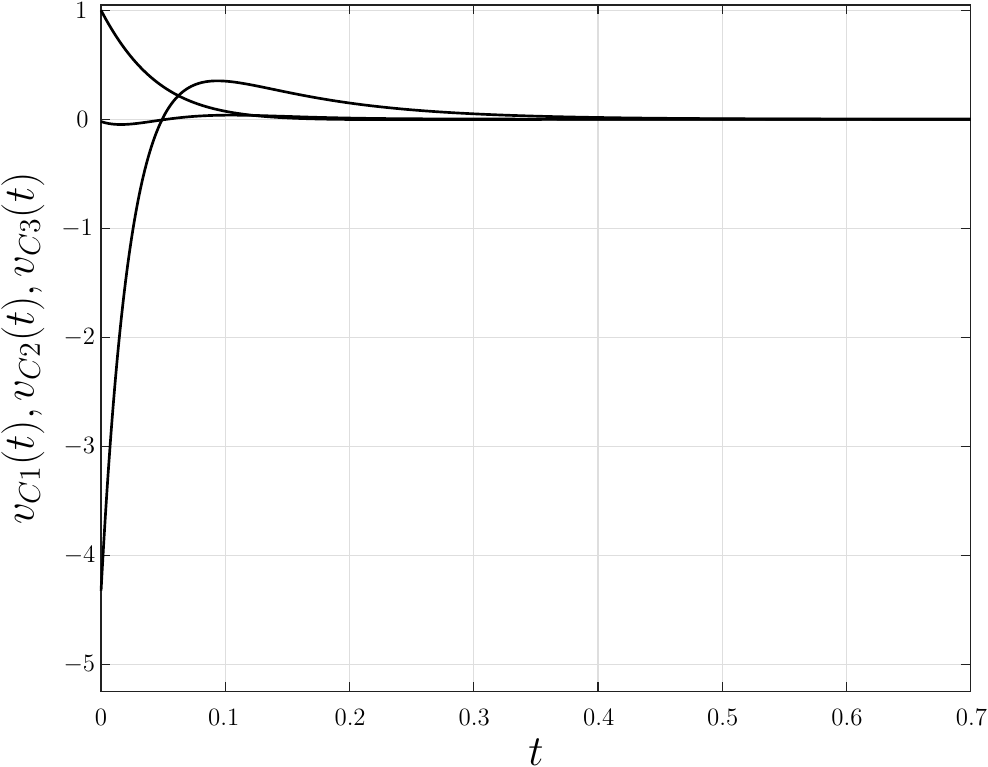}
  \caption{\small \small }
\end{subfigure}
\begin{subfigure}{0.49\columnwidth}
  \includegraphics[width=\linewidth]{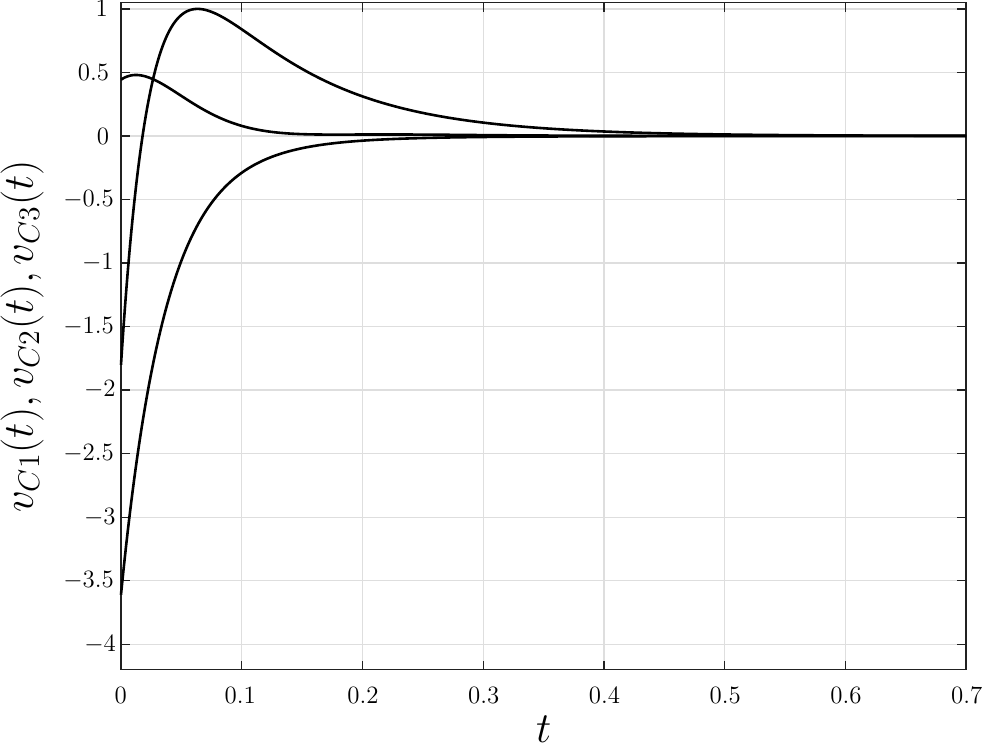}
  \caption{\small \small }
\end{subfigure}\\
\caption{\small Capacitor voltages $v_{Ci}$, $i=1,2,3$, for four solutions
of the memristor circuit in Example\ 1. Case $k_L=0.01$. Horizontal axis: time. Vertical axis: capacitor
voltages in
normalized values.}
\label{fig:voltagescirc1v}
\end{center}
\end{figure}

\begin{figure}[t]
\begin{center}
\begin{subfigure}{0.49\columnwidth}
  \includegraphics[width=\linewidth]{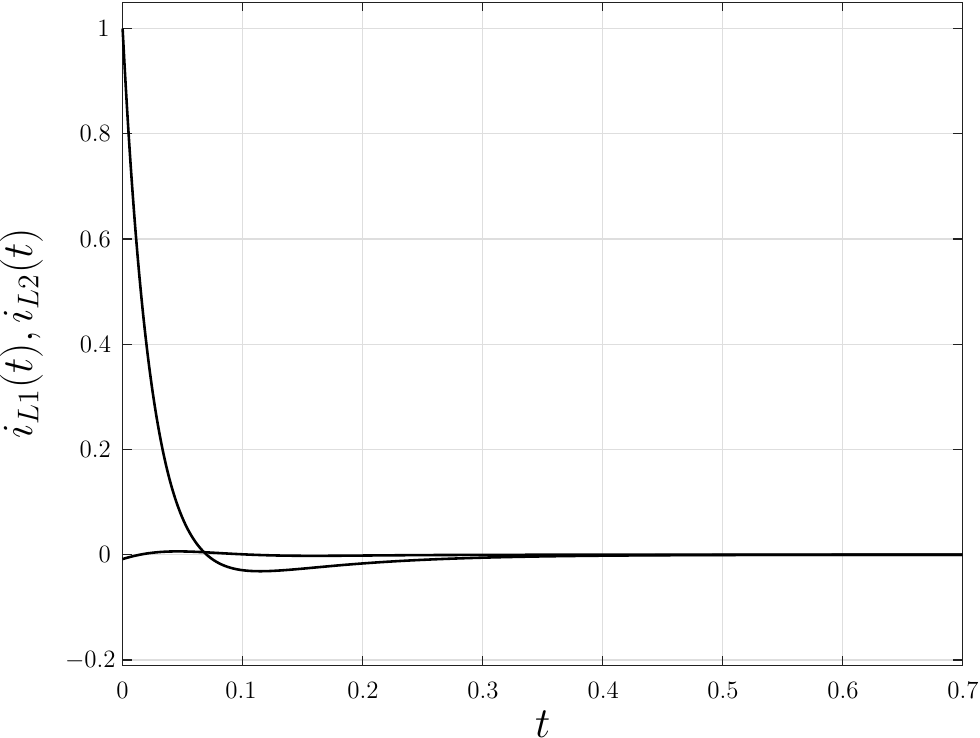}
  \caption{\small \small }
\end{subfigure}
\begin{subfigure}{0.49\columnwidth}
  \includegraphics[width=\linewidth]{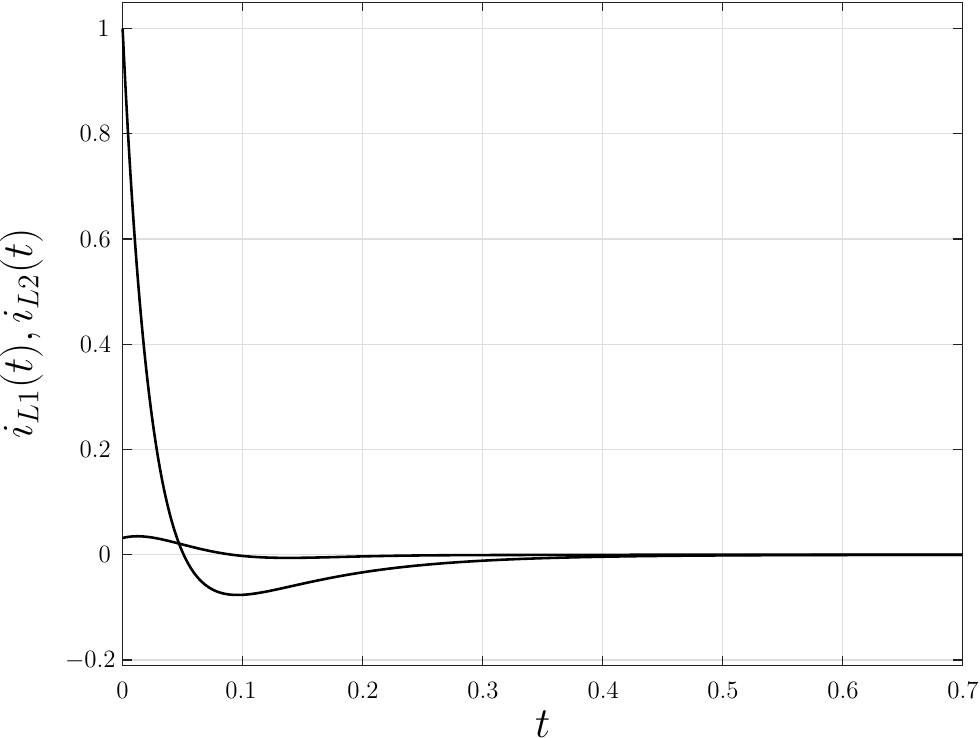}
  \caption{\small \small }
\end{subfigure} \\
\begin{subfigure}{0.49\columnwidth}
  \includegraphics[width=\linewidth]{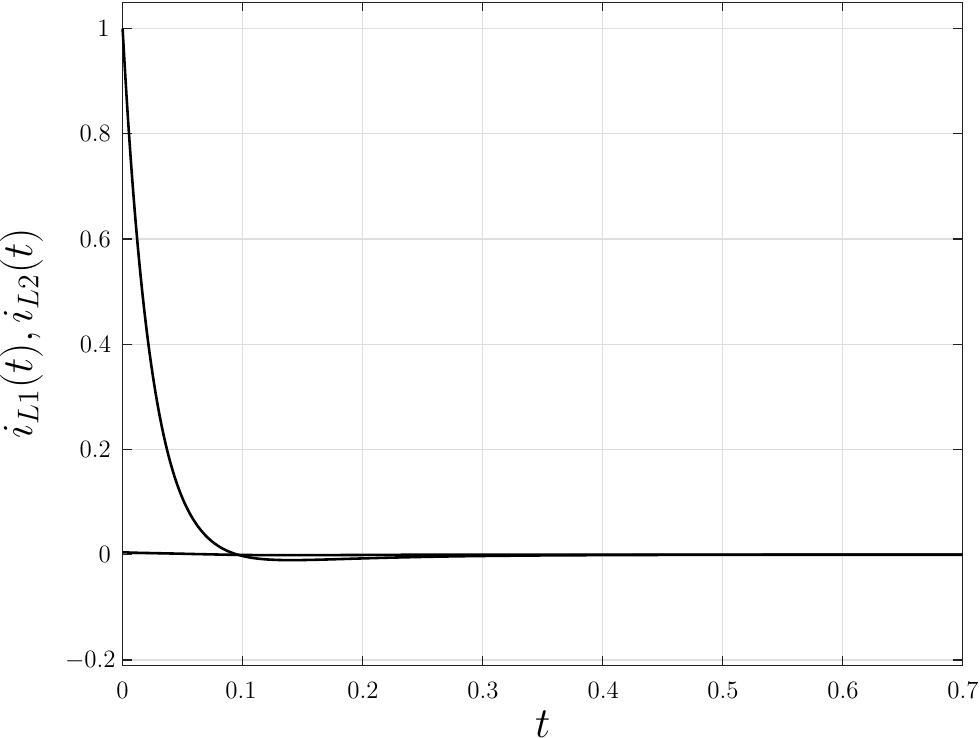}
  \caption{\small \small }
\end{subfigure}
\begin{subfigure}{0.49\columnwidth}
  \includegraphics[width=\linewidth]{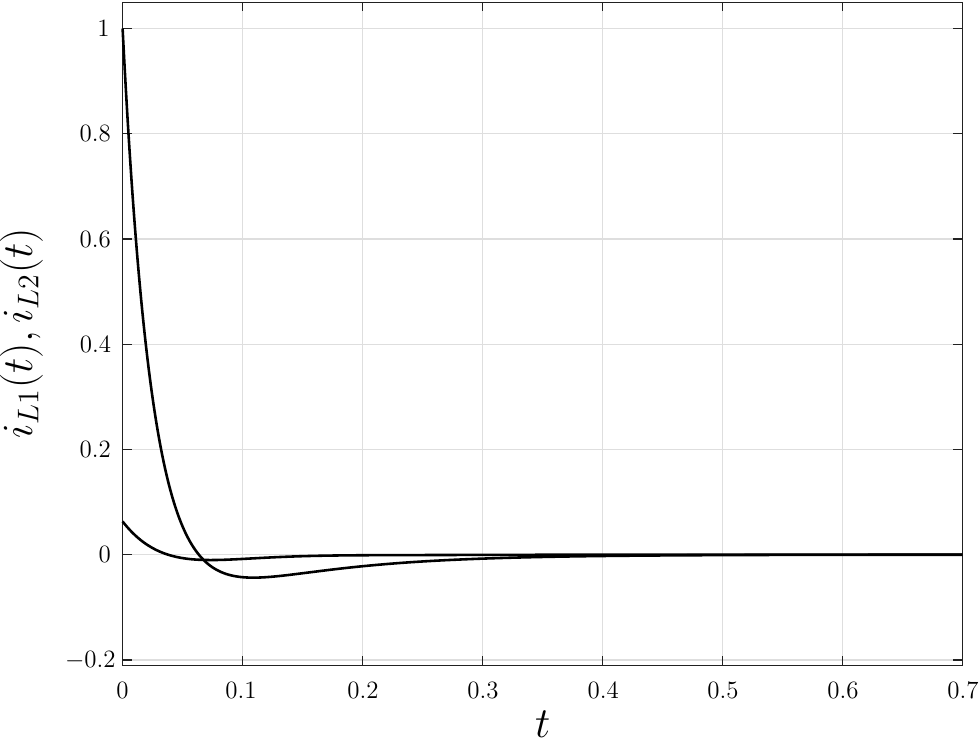}
  \caption{\small \small }
\end{subfigure}\\
\caption{\small Inductor currents $i_{Li}$, $i=1,2$, for four solutions
of the memristor circuit in Example\ 1. Case $k_L=0.01$. Horizontal axis: time. Vertical axis: inductor currents in
normalized values.}
\label{fig:currentscirc1i}
\end{center}
\end{figure}

\begin{figure}[t]
\begin{center}
\begin{subfigure}{0.49\columnwidth}
  \includegraphics[width=\linewidth]{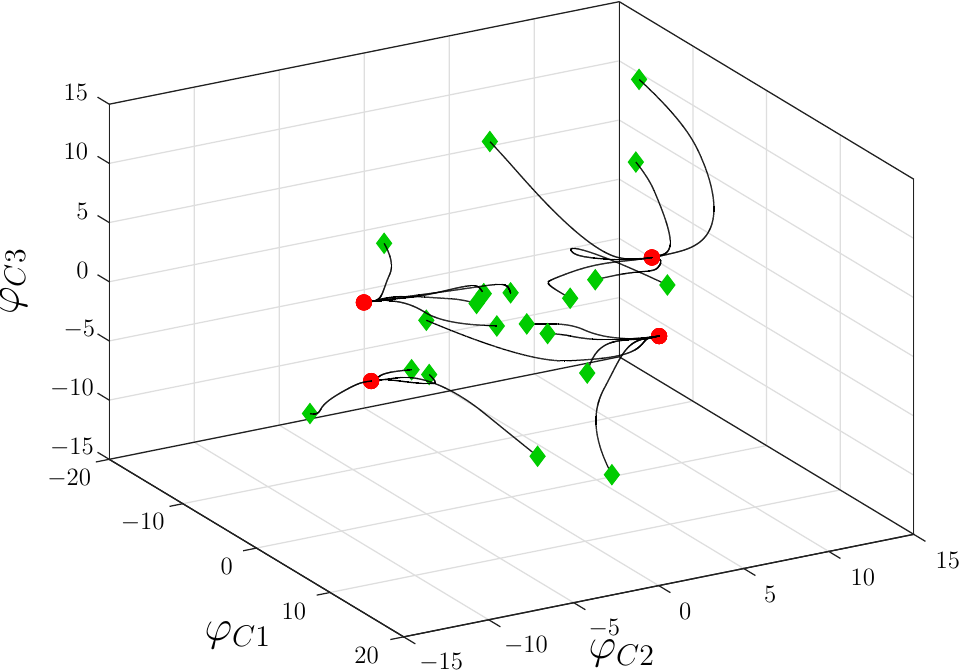}
  \caption{\small \small }
\end{subfigure}
\begin{subfigure}{0.49\columnwidth}
  \includegraphics[width=\linewidth]{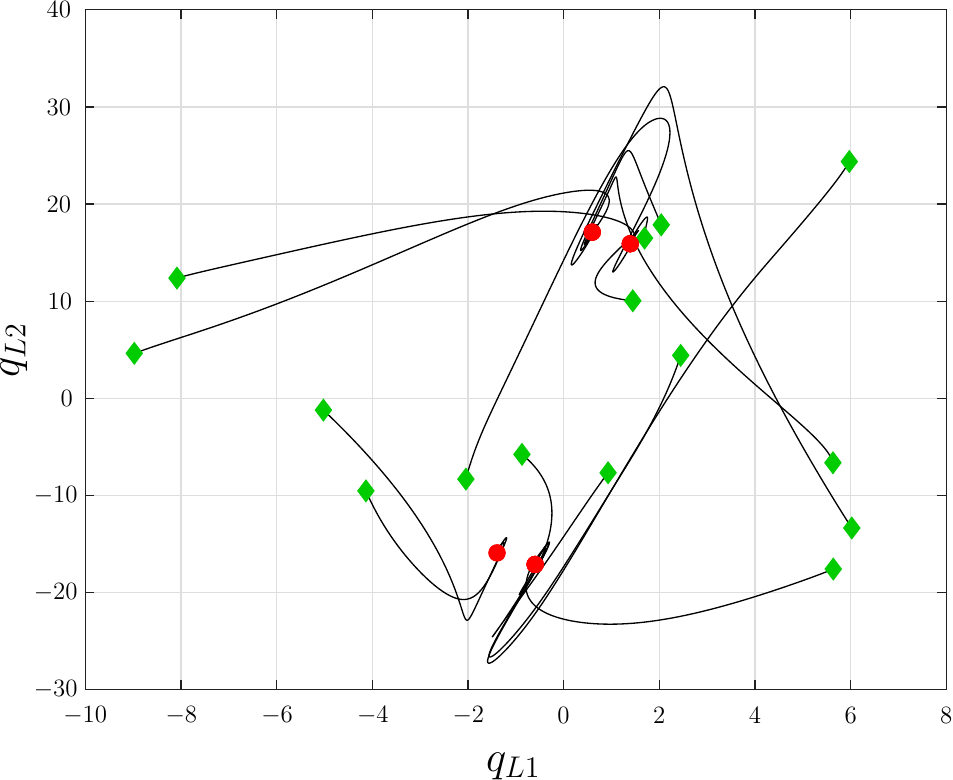}
  \caption{\small \small }
\end{subfigure} \\
\begin{subfigure}{0.49\columnwidth}
  \includegraphics[width=\linewidth]{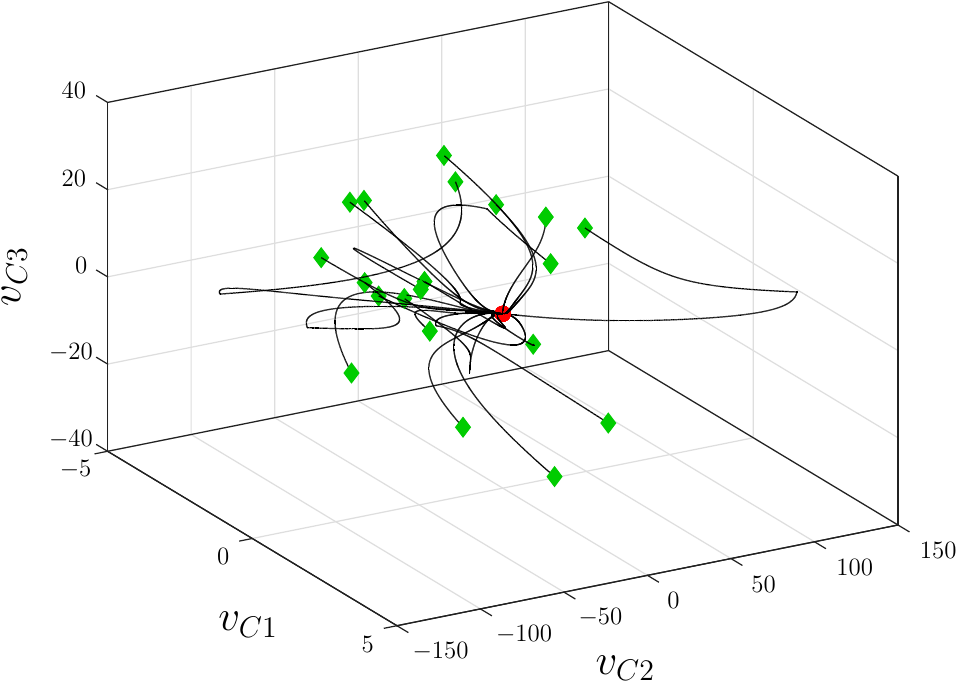}
  \caption{\small \small }
\end{subfigure}
\begin{subfigure}{0.49\columnwidth}
  \includegraphics[width=\linewidth]{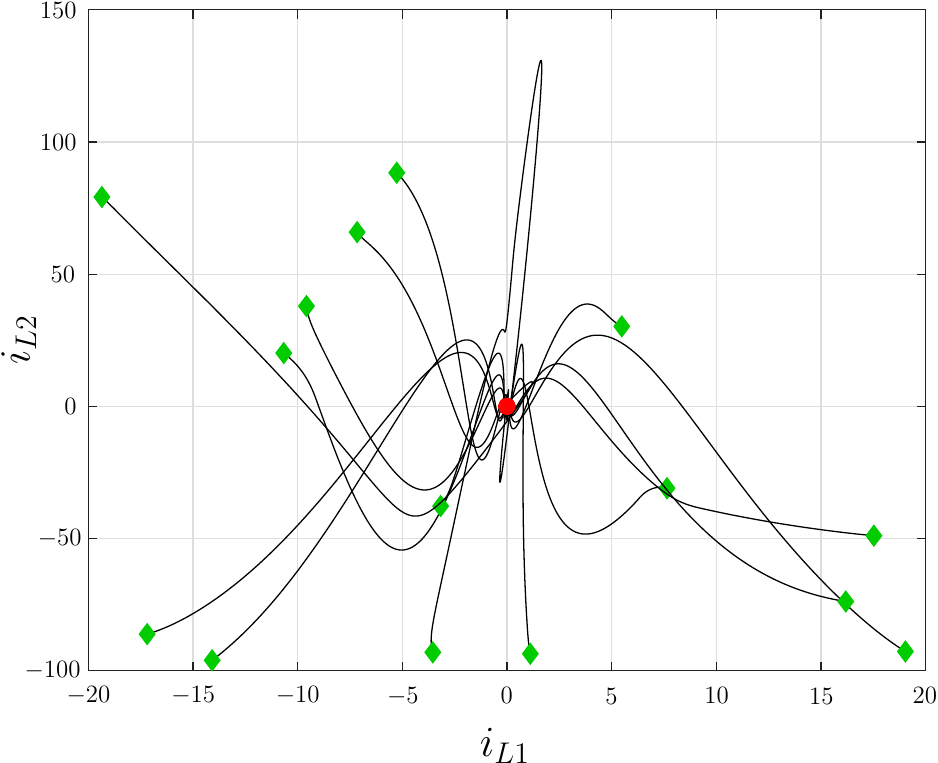}
  \caption{\small \small }
\end{subfigure}\\
\caption{\small Phase-space evolution of solutions starting at different initial conditions
for the memristor circuit in Example\ 1. Case $k_L=0.01$. Initial conditions are denoted by green points and
EPs by red points.
(a) Capacitor fluxes, (b) inductor charges,
(c) capacitor voltages and (d) inductor currents.}
\label{fig:statespacecirc1}
\end{center}
\end{figure}

\begin{figure}[t]
\begin{center}
\begin{subfigure}{0.49\columnwidth}
  \includegraphics[width=\linewidth]{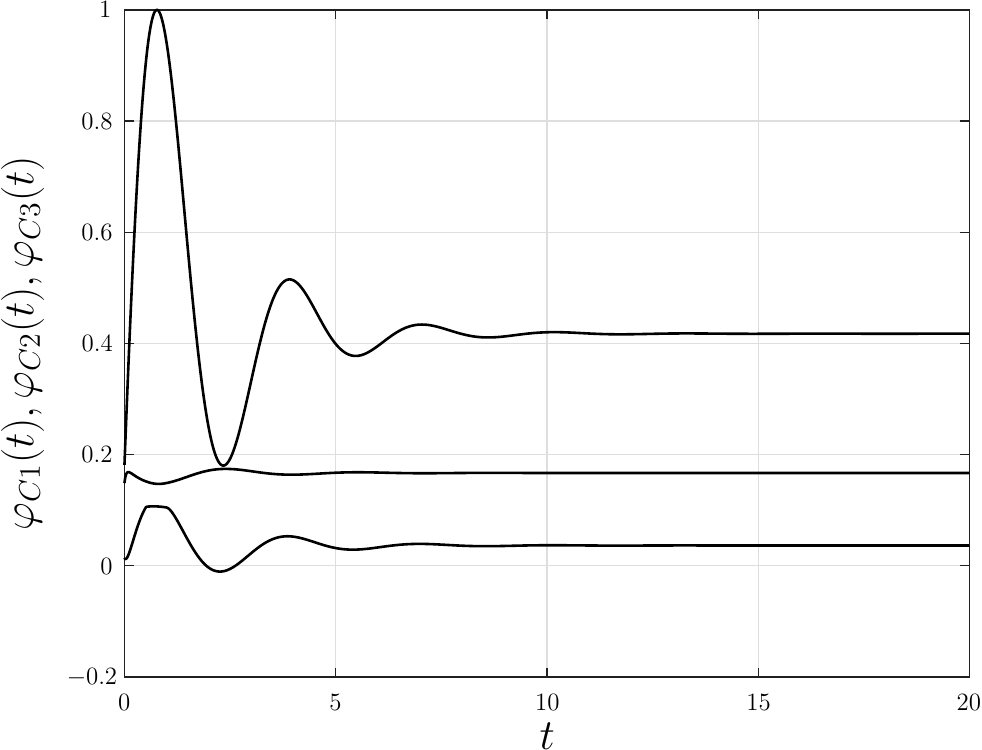}
  \caption{\small \small }
\end{subfigure}
\begin{subfigure}{0.49\columnwidth}
  \includegraphics[width=\linewidth]{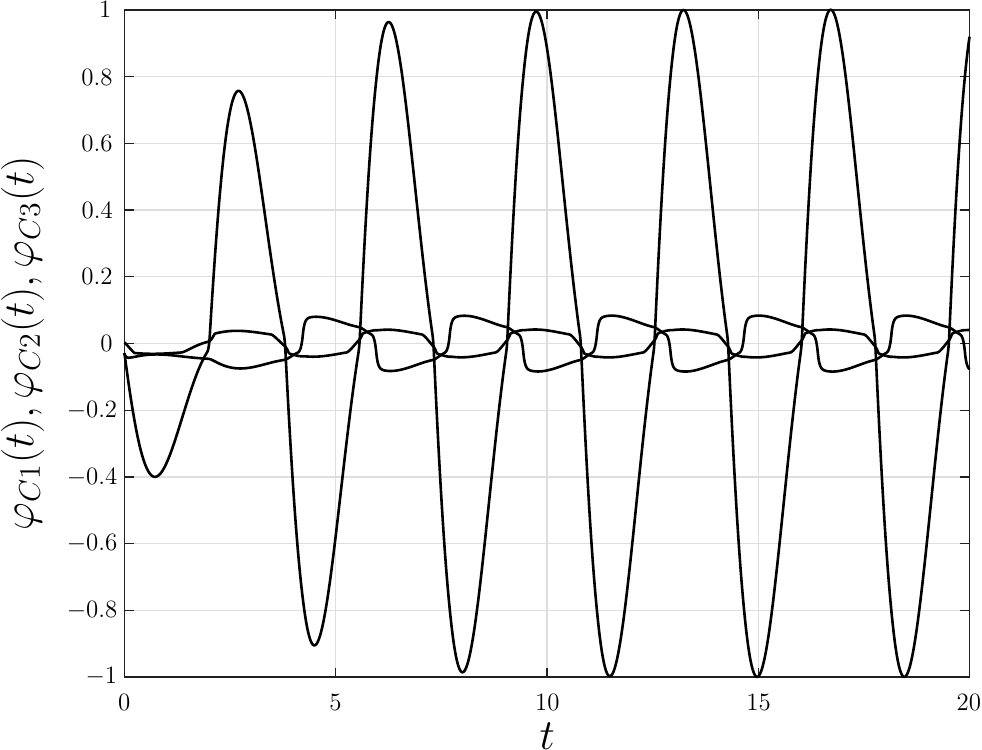}
  \caption{\small \small }
\end{subfigure} \\
\begin{subfigure}{0.49\columnwidth}
  \includegraphics[width=\linewidth]{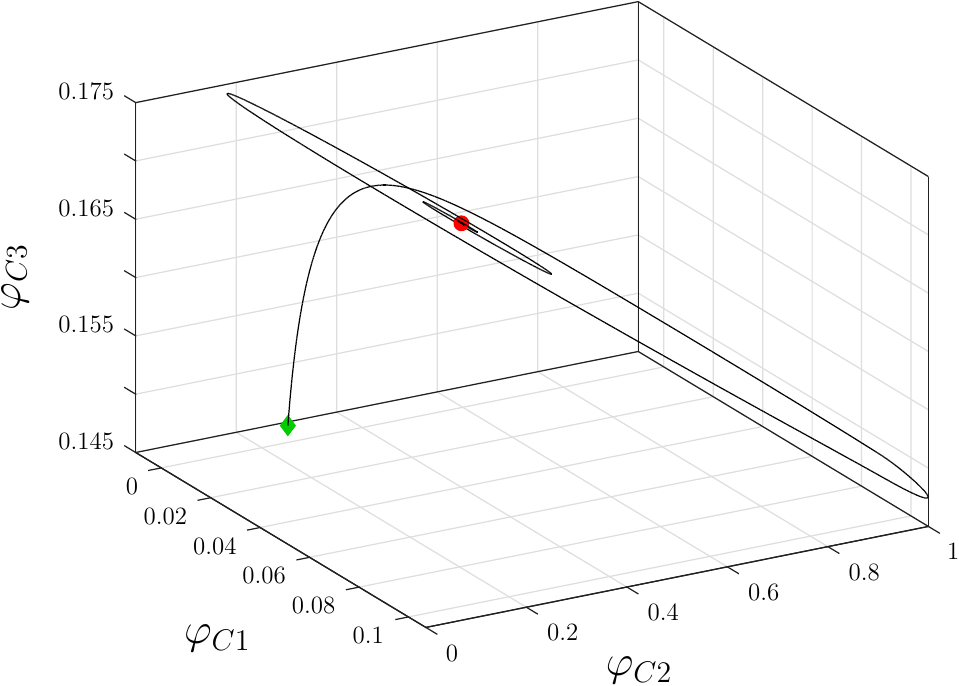}
  \caption{\small \small }
\end{subfigure}
\begin{subfigure}{0.49\columnwidth}
  \includegraphics[width=\linewidth]{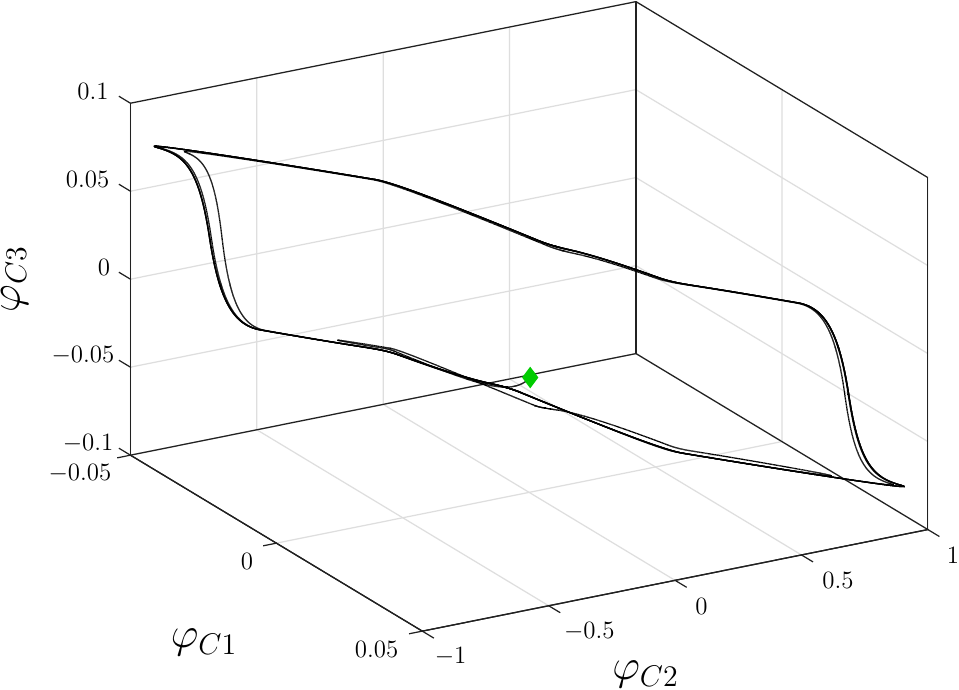}
  \caption{\small \small }
\end{subfigure}\\
\caption{\small Capacitor fluxes $\varphi_{Ci}(t)$, $i=1,2,3,$ of the memristor
circuit in Example\ 1. Case $k_L=1$. (a) Solution that displays vanishing oscillations
and converges to a stable EP. (b) A different solution that
displays non-vanishing oscillations and converges to a
stable limit cycle. (c) State-space
evolution of the solution in (a). (d) State-space evolution of the solution in (b).}
\label{fig:circ1oscill}
\end{center}
\end{figure}

\subsection{Example\ 2}

Consider the memristor circuit $\mathfrak{N}$ in the class RLCM$c$ shown in Fig.\ \ref{fig:fluxcirc2} (cf.\ Sect.\ \ref{sect:subc RLCMc}).
There are 3 capacitors, 2 inductors and 6 linear resistors, moreover, any capacitor has in parallel a flux-controlled memristor and any inductor has in series a charge-controlled memristor.
Choose $\cT=\{ C_1,C_2,C_3,R_4,R_5,R_6 \}$, $\cL=\{ L_1, L_2, R_1,R_2,R_3 \}$, $\cT'=\{ C_1, C_2, C_3 \}$,
$\cL'=\{ R_1,R_2,R_3 \}$, $\cT-\cT'=\{ R_4,R_5,R_6 \}$ and $\cL-\cL'=\{ L_1,L_2 \}$. It can be checked that Assumptions\ \ref{assu:case 2 RLCM 1}, \ref{assu:case 2 RLCM 1b} are satisfied, hence the capacitor fluxes together with the inductor charges are a complete set of variables in the FCD.
The topological matrix $A$ in (\ref{P partition}) is given by
$$
    A=\begin{pmatrix}
        A_{CL} & A_{CR} \\
        A_{RL} & A_{RR} \\
      \end{pmatrix}=
      \left(
      \begin{array}{c c | c c c}
        -1 & 0 & -1 & 0 & 0 \\
        -1 & 0 & 0 & -1 & 0 \\
        0 & -1 & 0 & 0 & -1 \\
        \hline
        1 & 1 & 0 & 0 & 0 \\
        1 & 0 & 0 & 0 & 0 \\
        0 & 1 & 0 & 0 & 0 \\
      \end{array}
      \right).
$$

Linear resistors belonging to $\cT-\cT'$ are chosen as $R_4=0.04,R_5=0.02,R_6=0.05$ Ohm, while
resistors in $\cL'$ are $R_1=-0.6,R_2=-0.8,R_3=-0.5$ Ohm.
On this basis, matrix $H$ is expressed as
\begin{equation}\label{Hex1}
    H\! \! = \! \! \begin{pmatrix}
        H_{\alpha \alpha} & H_{\alpha \beta} \\
        -H_{\alpha \beta}^\top & H_{\beta \beta} \\
      \end{pmatrix} \! \! = \! \!
      \left (
      \begin{array}{c c c | c c}
        -1.667 & 0 & 0 & -1 & 0 \\
        0 & -1.25 & 0 & -1 & 0 \\
        0 & 0 & -2 & 0 & -1 \\
        \hline
        1 & 1 & 0 & 0.06 & 0.04 \\
        0 & 0 & 1 & 0.04 & 0.09 \\
      \end{array}
      \right ).
\end{equation}

The capacitor values have been set to $C_1=C_2=C_3=1$ F,
while the inductors are $L_1=L_2=k_L$ H, where $k_L>0$
is a parameter.
Any flux-controlled memristor (resp., charge-controlled memristor)
has the same characteristic given by a $C^1$ approximation
of the piecewise-linear function $Q=\hat Q(\Phi)=3.5 \Phi -(2.5/2)
(|\Phi+2|-|\Phi-2|)$ (resp., $\Phi=\hat \Phi(Q)=3.5 Q -(2.5/2)
(|Q+2|-|Q-2|)$). Assumptions\ \ref{assu:memp}, \ref{assu:memc} are
satisfied. In
particular, we have $G_{mi}=G_\mathrm{\mathrm{m}}=3.5$ Ohm$^{-1}$, $i=1,2,3$
and $R_{mi}=3.5$ Ohm, $i=1,2$.

Matrix $H_{\beta \beta}$ is positive definite. For these values we have that
conditions (\ref{mu1}), (\ref{mu2}) in Theorem\ \ref{th_RLCMc} are satisfied
if $k_L< k_{L\rm{max}}=1.588$. Moreover, since $G_\mathrm{\mathrm{m}}=3.5
> \| H_{\alpha \alpha} , H_{\alpha \beta} \|_\infty=3$ and
$
R_\mathrm{\mathrm{m}}=3.5> \|- H_{\alpha \beta}^\top  ,  H_{\beta \beta}  \|_\infty
=2.1,
$
condition (\ref{RLMC_condition1}) in Proposition\ \ref{th_RLCMc-bound}
is satisfied.

We have simulated with MATLAB the dynamic behavior for different values of $k_L$.
All results have been in agreement with those predicted by the theory.
Below, we briefly illustrate the simulations in the case $k_L=1.1<k_{L\rm{max}}$
and the manifold index is $(Q_0,\Phi_0)=(0,0)$.
Via a suitable MATLAB program, we have found that the number of EPs of
the SEs (\ref{SEs FCD RLCM special}) is equal to 9, moreover, there are
4 stable EPs and 5 unstable EPs. Figure\ \ref{fig:statespacecirc2}
shows the behavior of a number of convergent solutions starting at different
initial conditions in the phase space in the FCD. As predicted by
Proposition\ \ref{th_RLCMc-bound} and Theorem\ \ref{th_RLCMc}, all solutions
are bounded and they converge to an EP.

\begin{figure}[t]
\begin{center}
\begin{subfigure}{0.49\columnwidth}
  \includegraphics[width=\linewidth]{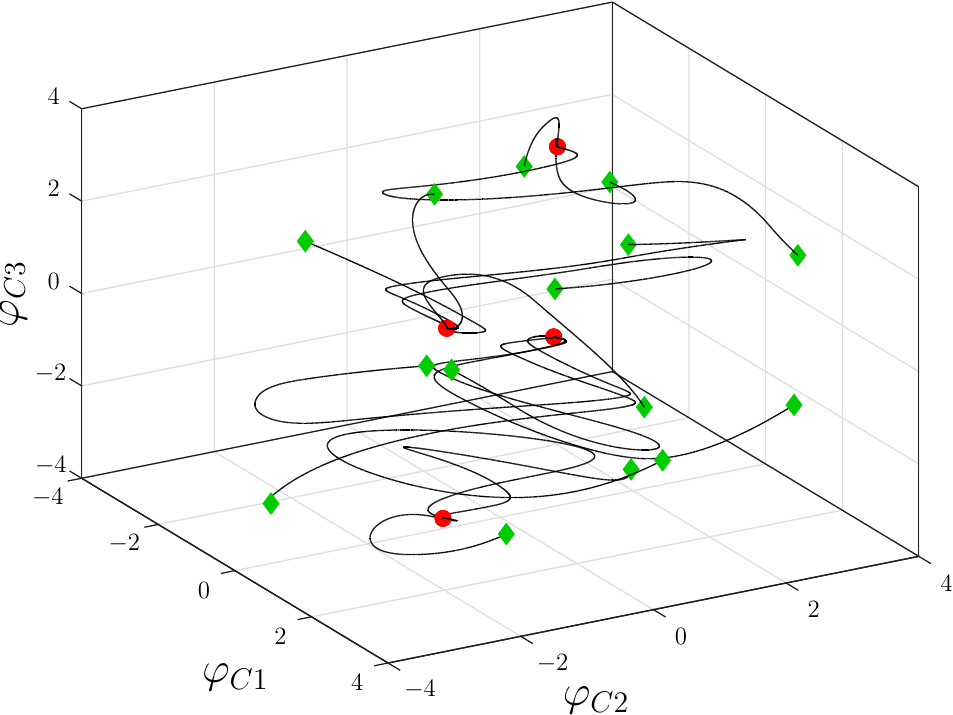}
  \caption{\small \small }
\end{subfigure}
\begin{subfigure}{0.49\columnwidth}
  \includegraphics[width=\linewidth]{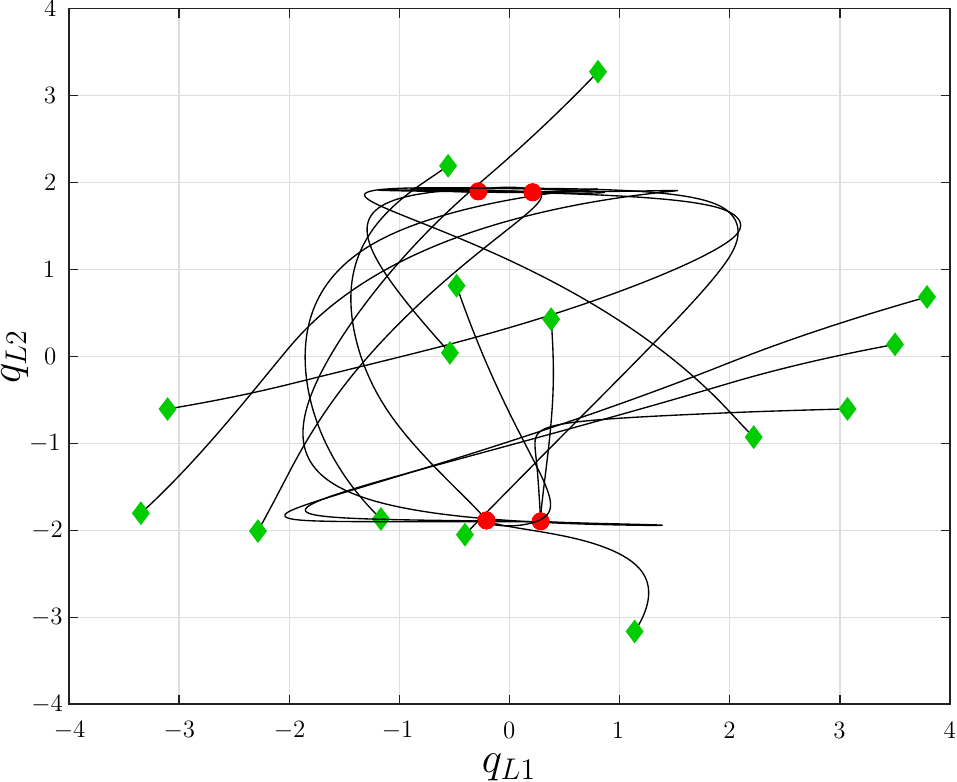}
  \caption{\small \small }
\end{subfigure} \\
\caption{\small State-space evolution of solutions starting at different initial conditions
for the memristor circuit in Example\ 2. Case $k_L=1.5$. Initial conditions are denoted by green points and
EPs by red points.
(a) Capacitor fluxes, (b) inductor charges.}
\label{fig:statespacecirc2}
\end{center}
\end{figure}

\section{Conclusion}
\label{sect:concl}
The paper has considered a large class of nonlinear circuits, termed RLCM, containing
all four basic circuit elements, i.e., resistors, inductors, capacitors and memristors.
By using the mixed potential introduced in a companion paper \cite{DiMarco2026BraytonTheor},
systematic results
on convergence for RLCM circuits have been established. The results generalize
previous results on convergence for nonlinear RLC circuits without memristors
and nonlinear memristor circuits without inductors and, to the authors
knowledge, they are the only existing convergence results for
circuits containing all four basic circuit elements. The obtained convergence results
are robust, i.e., they hold in an open set of circuit parameters and for typical
nonlinearities used to model the memristors. Conditions
to have multiple stable EPs are clearly identified based on the use and location
of active resistors. Multiple stable EPs are of importance in several engineering
applications as the implementation of CAMs. The results are illustrated via the
application to some specific RLCM circuits. Future extensions will address the
possibility to relax some of the assumptions made in the convergence theorems in
order to broaden the field of applicability. Further work will also be devoted to the
application to high-dimensional circuits with a neural-like architecture in view
of the implementation of CAMs.

%
%
%

\section*{Appendix A}
Function $V(Z(\cdot))$ is convex and Lipschitz with Lipschitz constant $1$. Then, by the chain rule as in \cite[Property\ 1]{FGNP06} we have
\begin{align}
\label{eq:scalarprod}
\dot V(Z(t)) & = \langle \xi,  Z(t) \rangle, \ \ \forall \xi \in \partial V (Z(t))
\end{align}
for almost all $t\in [0,\tau]$, where $\partial V(Z(t))$ is the subgradient in the sense of convex analysis of $V$ evaluated at $Z(t)$ \cite[Def.\ 1.2.1]{convex} and $\langle \cdot , \cdot \rangle$ is the inner product.
Recall that $k(t) \in \{1,\dots, n_C+n_L\}$ is such that $|Z_{k(t)}(t)|=\|Z(t)\|_\infty$ for $t \geq 0$. It can be checked that, for any $Z(t) \neq 0$, vector $\xi(t) \in \R^{n_C+n_L}$ such that $\xi_{k(t)}(t) = \mathrm{sgn}(Z_{k(t)}(t))$ and $\xi _i(t) = 0$ for any $i \in \{1,\dots, n_C+n_L\}$, $i \neq k(t)$, is such that $\xi(t) \in \partial V (Z(t))$.
Then, using~(\ref{SEs FCD RLCM special}) to evaluate~(\ref{eq:scalarprod}) we obtain
\begin{equation}
\label{eq:valuescalprodC}
\begin{split}
\dot V(Z(t))=&-\hat{Q}_{Mk(t)}(Z_{k(t)}(t)) \, \mathrm{sgn} (Z_{k(t)}(t))\\
 &+ Q_{0k(t)} \, \mathrm{sgn} (Z_{k(t)}(t)) \\
& - \left(\sum_{j=1}^{n_C+n_L}h_{k(t)j} Z_j(t)\right) \mathrm{sgn} (Z_{k(t)}(t))
\end{split}
\end{equation}
if $k(t) \in {1,\dots,n_C}$, whereas
\begin{equation}
\label{eq:valuescalprodL}
\begin{split}
\dot V(Z(t))= &- \hat{\Phi}_{M(k(t)-n_C)}( Z_{k(t)}(t)) \,
 \mathrm{sgn} (Z_{k(t)}(t))\\
 &+ \Phi_{0(k(t)-n_C)} \, \mathrm{sgn} (Z_{k(t)}(t))\\
 &- \left(\sum_{j=1}^{n_C+n_L}h_{k(t)j}Z_j(t)\right) \mathrm{sgn} (Z_{k(t)}(t))
\end{split}
\end{equation}
if $k(t) \in {n_C+1,\dots,n_C+n_L}$.
Let us first focus on the case $k(t) \in {1,\dots,n_C}$ (cf.~(\ref{eq:valuescalprodC}))
and assume that
\begin{equation}
\label{eq:conditionZC}
|Z_{k(t)}(t)| > \max_{i=1,\dots,n_C}\{ \tilde \Phi_{Mi} \}.
\end{equation}
First, note that
\begin{align*}
Q_{0k(t)} \mathrm{sgn} (Z_{k(t)}(t)) \le  \| Q_0\|_\infty .
\end{align*}
Moreover, since $|Z_{k(t)}| \ge |Z_j|$, it holds
\begin{align*}
-& \left(\sum_{j=1}^{n_C+n_L} h_{k(t)j} Z_j(t)\right) \mathrm{sgn} (Z_{k(t)}(t)) \\
\leq & \sum_{j=1}^{n_C+n_L} |h_{k(t)j}| |Z_k(t)| =  \| H_{\alpha \alpha}, H_{\alpha \beta} \|_\infty \|Z_{k(t)}(t)\|_\infty.
\end{align*}
Finally, considering that $\hat Q_{Mi}(0)=0$, $\hat Q'_{Mi}(\Phi_{Mi})\geq 0$, and conditions~(\ref{Gmi}) holds, then, under condition~(\ref{eq:conditionZC}), we have
\begin{align*}
&\hat{Q}_{Mk(t)}\left(Z_{k(t)}(t) \right) \mathrm{sgn} (Z_{k(t)}(t)) =  \left| \hat{Q}_{Mk(t)}\left( Z_{k(t)}(t)  \right) \right| \\
& \geq G_{Mk(t)} |Z_{k(t)}(t)| = G_{Mk(t)} \|Z_{k(t)}(t)\|_\infty.
\end{align*}
Using the last three inequalities in~(\ref{eq:valuescalprodC}), we obtain
\begin{align*}
\begin{split}
\dot V(Z(t))
\leq &  - \| Z_{k(t)}(t)\|_\infty \\
 & \cdot \left( G_{Mk(t)} - \frac{ \| Q_0\|_\infty}{ \|Z_{k(t)}(t)\|_\infty} -\| H_{\alpha \alpha}, H_{\alpha \beta} \|_\infty  \right).
\end{split}
\end{align*}
By a similar argument, we can tackle the case $k(t) \in n_C+1,\dots,n_C+n_L$, under the assumption
$$
|Z_{k(t)}(t)| > \max_{i=1,\dots, n_L} \{\tilde Q_{Mi}\}.
$$
This completes the proof. \qed

\bibliographystyle{unsrt}

\end{document}